\documentclass[a4paper]{article}
\usepackage{times}
\usepackage{a4wide}
\usepackage{hyperref}
\usepackage{authblk}
\usepackage{url}
\usepackage{amsmath,amsfonts, amsthm,amssymb}
\usepackage{algorithmic}
\usepackage{textcomp}
\usepackage{xcolor}
\usepackage{svg}
\usepackage{lscape}
\AtBeginDocument{%
	\providecommand\BibTeX{{%
			\normalfont B\kern-0.5em{\scshape i\kern-0.25em b}\kern-0.8em\TeX}}}

\usepackage{tikz,ifthen}                       
\usetikzlibrary{shapes.geometric, decorations, calc, math}
\usepackage{pgfplots}
\pgfplotsset{
	compat=newest,
}
\usepackage{natbib}

\usepackage{latexsym}
\usepackage{caption, subcaption}
\usepackage{enumitem}
\usepackage{graphicx}
\usepackage{enumitem}

\usepackage[ruled, noend, linesnumbered]{algorithm2e}
\usepackage{comment}
\usepackage{booktabs}
\usepackage{cleveref}

\newtheorem{theorem}{Theorem}

\newtheorem{problem}{Problem}

\newtheorem{proposition}{Proposition}
\newtheorem{lemma}{Lemma}
\newtheorem{corollary}{Corollary}

\definecolor{UniBlue}{HTML}{07529A}
\definecolor{UniOrange}{HTML}{EAB90C}
\definecolor{UniGray}{HTML}{909085}

\newcommand{\opn}[1]{\operatorname{#1}}

\newcommand{\bigO}[1]{O \left( #1 \right)}

\newcommand{\plane}{\mathbb{R}^2}

\newcommand{\cc}{\rho}

\newcommand{\join}[1]{\sup}

\newcommand{\meet}[1]{\inf}

\newcommand{\tG}{{\widetilde{G}}}
\newcommand{\degree}[1]{d(#1)}

\newcommand{\ccG}{\rho_G}

\newcommand{\facenumber}{\Phi}
\newcommand{\closuretree}{\tau}
\newcommand{\closureblock}{\beta}
\newcommand{\generatorset}{\text{\sc GeneratorSet}}

\newcommand{\last}{\uparrow}
\newcommand{\lastL}{{\last}_{L}}

\newcommand{\lastR}{{\last}_{R}}
\newcommand{\reach}{\opn{reach}}
\newcommand{\bound}{\sigma}
\newcommand{\Left}{\mathcal{L}}
\newcommand{\Right}{\mathcal{R}}
\newcommand{\Leftnew}{\mathcal{L}^\text{new}}
\newcommand{\Rightnew}{\mathcal{R}^\text{new}}
\newcommand{\dR}{L}
\newcommand{\dL}{R}
\newcommand{\Path}{\opn{Path}}
\newcommand{\ER}{Erd\H{o}s-Rényi\xspace}
\newcommand{\Core}{\mathcal{C}}

\DeclareMathOperator*{\argmax}{arg\,max}

\newlength{\somewidth}
\newlength{\someheight}
\let\oldnl\nl
\newcommand{\nonl}{\renewcommand{\nl}{\let\nl\oldnl}}

\title{A Fast Heuristic for Computing Geodesic Cores\\ in Large Networks}

\author[1]{Florian Seiffarth}
\author[1,2,3]{Tam\'{a}s Horv\'{a}th}
\author[1,2,3]{Stefan Wrobel}

\affil[ ]{\footnotesize\textit {\{seiffarth, horvath, wrobel\}@cs.uni-bonn.de}}
\affil[1]{\footnotesize Dept. of Computer Science, University of Bonn, Bonn,
	Germany}
\affil[2]{\footnotesize Fraunhofer IAIS, 
	Schlo{\ss} Birlinghoven, Sankt Augustin, Germany}
\affil[3]{\footnotesize Fraunhofer Center for Machine Learning, 
	Schlo{\ss} Birlinghoven, Sankt Augustin, Germany}	

\begin{document}

\maketitle
	
\begin{abstract}
	Motivated by the increasing interest in applications of graph geodesic convexity in machine learning and data mining, we present a heuristic for computing the geodesic convex hull of node sets in networks. 
	It generates a set of almost maximal outerplanar spanning subgraphs for the input graph, computes the geodesic closure in each of these  graphs, and regards a node as an element of the convex hull if it belongs to the closed sets for at least a user specified number of outerplanar graphs. 
	Our heuristic algorithm runs in time linear in the number of edges of the input graph, i.e., it is faster with one order of magnitude than the standard algorithm computing the closure exactly.
	Its performance is evaluated empirically by approximating convexity based core-periphery decomposition of networks.
	Our experimental results with large real-world networks show that for most networks, the proposed heuristic was able to produce close approximations significantly faster than the standard algorithm computing the exact convex hulls. For example, while our algorithm calculated an approximate core-periphery decomposition in 5 hours or less for networks with more than 20 million edges, the standard algorithm did not terminate within 50 days. 
\end{abstract}

\maketitle


\section{Introduction}	
\label{sec:intro}
In recent years, there has been a growing interest in applications of \textit{geodesic convexity} in graphs (see, e.g., \cite{Pelayo13}).
This concept has been utilized successfully also in \textit{machine learning} and \textit{data mining}, besides other fields of computer science (e.g., genome rearrangement problems~\cite{Cunha18}).
Examples include exact \textit{cluster} recovery with queries~\cite{Bressan21}, 
vertex classification in batch~\cite{Macedo19,Seiffarth19, Stadtlander21} and active \textit{learning}~\cite{Thiessen/Gaertner/2021}, or \textit{mining} complex networks~\cite{Sub18, Subelj_etal/2019}.
Regarding this latter application, a \textit{new} type of \textit{core-periphery} network decomposition~\cite{Borg99} based on geodesic convexity has been proposed in \cite{Sub18}.
The results in \cite{Sub18}, as well as in subsequent papers \cite{Subelj_etal/2019, Sub19} clearly demonstrate that geodesic convexity based core-periphery   decomposition provides further useful insights into the network's structure, which have not been captured before. 
More precisely, by means of geodesic convexity, a broad class of real-world networks  can be decomposed into a dense core surrounded by a sparse periphery (see Fig.~\ref{fig:core_periphery} for a relatively small example).
\begin{figure}[t]
	\begin{subfigure}[t]{0.32\linewidth}
		\centering
		\includegraphics[scale=0.4]{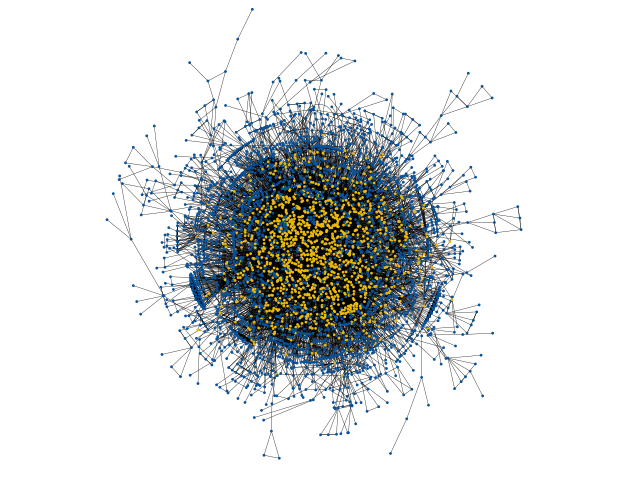}
		\caption{Entire network}
		\label{fig:Graph}
	\end{subfigure}
	\hfill
	\begin{subfigure}[t]{0.32\linewidth}
		\centering
		\includegraphics[scale=0.4]{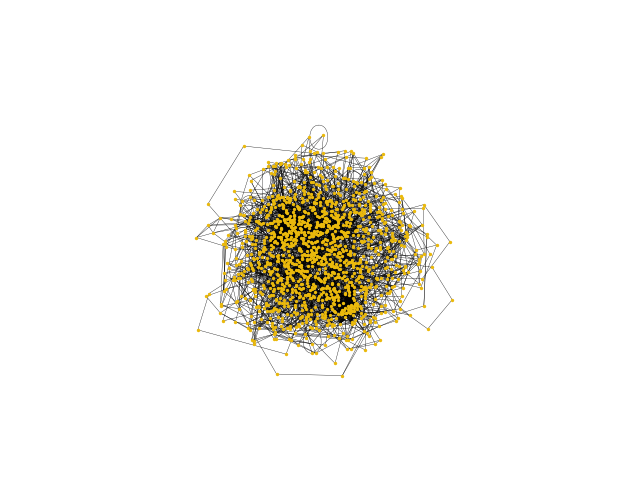}
		\caption{Geodesic core}
		\label{fig:Core1}
	\end{subfigure}
	\hfill
	\begin{subfigure}[t]{0.32\linewidth}
		\centering
		\includegraphics[scale=0.4]{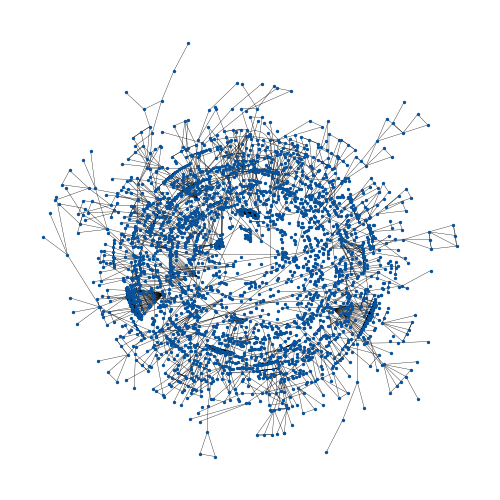}
		\caption{Periphery}
		\label{fig:Periphery1}
	\end{subfigure}
	\caption{\label{fig:core_periphery}(a) CA-GrQc network, (b) its (geodesic) core, (c) its periphery}
\end{figure}
In contrast to the core, the shortest paths between most node pairs in the periphery are unique. 
As mentioned, such a decomposition enables the acquirement of new knowledge~\cite{Sub18}.
For example, basically the nodes in the core govern the degree distribution of the entire network or they have a  higher clustering coefficients and a smaller geodesic distance to each other than the periphery nodes.
A further nice property of convexity based core-periphery decomposition is that it is \textit{not} characteristic to all network types. In particular, while for example social networks typically possess this kind of decomposition, this is not the case for standard random graph and network models, such as the Erd\H{o}s-R\'enyi, Watts-Strogatz, or Barab\'asi-Albert models (see \cite{Sub18}).

This and other applications of geodesic convexity rely on the basic operation of computing the \textit{geodesic convex hull} of a set $X$ of vertices of a graph, that is, the \textit{smallest} set $C$ of vertices containing $X$ such that \textit{all} vertices of \textit{all} shortest paths with both endpoints in $C$ belong also to $C$.
Such a smallest $C$ always exists and it is unique.
We will refer to the geodesic convex hull of a set of nodes as its \textit{closure}.
It is a folklore result that the closure of a set of nodes of a network with $n$ nodes and $m$ edges can be calculated in $\bigO{nm}$ time (cf.~\cite{Pelayo13}).
This time complexity, which is \textit{cubic} in $n$ in the worst case, makes all approaches relying on geodesic convex hulls practically \textit{infeasible} for \textit{large} networks.

To overcome this problem, we give up the demand for \textit{correctness} and calculate only an \textit{approximation} of the closure of $X$, by noting that we are not aware of any other existing closure approximation heuristic.   
More precisely, we generate a set of random spanning subgraphs of $G$, compute the closure of $X$ in these subgraphs separately, and regard a vertex of $G$ as an element of the convex hull of $X$ iff it belongs to the closure of $X$ in at least a user specified percentage of the number of spanning subgraphs.
The main question for this scheme is the choice of the class of the spanning subgraphs.
At first glance, forests might seem a natural candidate	 for their nice algorithmic properties.
However, a closer look at the problem as well as our empirical results reveal that already for graphs that are structurally very close to forests, a very poor approximation performance can be obtained in this way. 
This is because spanning forests may drastically distort shortest paths
even in sparse graphs.

Instead of forests, we therefore consider the class of \textit{outerplanar graphs}~\cite{Chartrand/Harary/1967} for spanning subgraphs for their nice algorithmic properties.
Although this class is only slightly beyond that of forests in the structural hierarchy,
our empirical results with large real-world networks show that for most networks, a close approximation of the geodesic convex hull can be obtained with outerplanar spanning subgraphs.
Our main contributions are threefold.
We (i) present a \textit{fast} and \textit{easy to implement} algorithm computing an almost inclusion \textit{maximal} outerplanar spanning subgraph in $\bigO{m}$ time, (ii) give an algorithm computing the (geodesic) convex hull in an \textit{outerplanar} graph $G$ in time \textit{linear} in its number of vertices and face number, where the face number is the maximum number of interior faces over the biconnected components of $G$, and
(iii) 
report experimental results with \textit{large} real-world networks which show that their cores (and hence, peripheries) can be approximated closely with this scheme in  \textit{feasible} time.
To the best of our knowledge, our approach and the algorithms in this work are all new. 
Furthermore, (i) and (ii) above may be of some independent interest as well.

Regarding (i), our algorithm
relies on the paradigm used e.g. in \cite{Fray2012}).
In particular, it generates first a DFS tree and adds then iteratively further edges to the left or right side of the paths in the tree in such a way that the resulting graph remains outerplanar.
Our  experimental results clearly show 
that the spanning outerplanar graphs generated with this algorithm
are \textit{nearly maximal} (less than $0.3\%$ of the edges were missing for maximality in the experiments).

Regarding~(ii), the closure of a set $X$ of vertices of an outerplanar graph $G$ with $n$ vertices can be computed in $\bigO{n|X|}$ time by solving the single source shortest path (SSSP) problem for all vertices in $X$.
Although this is already a significant improvement over the $\bigO{nm}$ bound for arbitrary graphs,
we calculate the closure with a more sophisticated algorithm.
One of the strengths  of this algorithm is that its complexity is \textit{independent} of the cardinality of $X$.
More precisely, it computes the closure \textit{correctly} and in $\bigO{nf}$ time, where $f$ is the \textit{face number} of $G$.
Although $f = \bigO{n}$ in the worst-case, our runtime results clearly reveal that this difference is essential in practice because $f$ is typically \textit{negligible} w.r.t. $n$ (e.g., in case of outerplanar spanning subgraphs of Erd\H{o}s-R\'enyi random graphs with around 1,000,000 edges,
the average face number was consistently less than $80$).

Finally, regarding~(iii), our experiments with 15 \textit{large} real-world networks show that their cores can be approximated  with our heuristic algorithm with a Jaccard similarity between 82 and 99\%.
Furthermore, already 100 spanning outerplanar subgraphs sufficed to obtain a close approximation in all of our experiments.
Thus, our algorithm is \textit{linear} in the number of edges in \textit{practice}, in contrast to the $\bigO{nm}$ worst-case complexity of the standard algorithm based on SSSP.

In particular, in case of networks with more than 20 million (and up to 117 million) edges, the approximate decomposition could be computed in 5 hours or less  with our algorithm. In contrast, the computation of the exact core-periphery decomposition with the standard algorithm had to be aborted after 50 days.
Regarding the approximate cores of the networks, their degree distributions were close to those of the exact ones.

The rest of the paper is organized as follows. In Sect.~\ref{sec:preliminaries} we collect the necessary notions and fix the notation. Sections~\ref{sec:sampling} resp. \ref{sec:closure} contain the description of the algorithms generating spanning outerplanar subgraphs resp. computing the closure of a set of vertices in outerplanar graphs.
In Sec.~\ref{sec:experiments} we empirically evaluate our approach.
Finally, in Sect.~\ref{sec:conclusion} we formulate some questions for further research.

\section{Notions and Notation}	
\label{sec:preliminaries}

For basic notions in graph theory, we refer e.g. to \cite{Diestel/2012}.
The set $V$ of vertices (resp. $E$ of edges) of a graph $G= (V,E)$ is denoted by $V(G)$ (resp. $E(G)$).
By graphs we always mean finite \textit{undirected} and \textit{unweighted} graphs without loops and parallel edges that are \textit{connected} and denote $|V(G)|$ and $|E(G)|$ by $n$ and $m$, respectively.
Furthermore, for $u,v \in V(G)$, $uv$ stands for the edge $\{u,v\} \in E(G)$.  
The path connecting two nodes $u,v$ of a tree is denoted by $\Path(u,v)$.

Given a graph $G$, the function $I: V\times V\rightarrow 2^V$, called the \textit{geodesic interval}, maps $(u,v)$ to the union of the sets of vertices on all shortest paths between $u$ and $v$. 
A set $X\subseteq V(G)$ is \textit{(geodesically) convex} or \textit{closed} if $I(u, v)\subseteq X$ for all $u, v\in X$. 
For all $X\subseteq V(G)$, there exists a unique smallest closed set $X' \supseteq X$, called the \textit{convex hull} or \textit{closure} of $X$.
Furthermore, the function $\cc_G$ mapping the subsets of $V(G)$ to their convex hulls is a \textit{closure operator}, i.e.,  
it is 
\textit{extensive} (i.e., $X\subseteq \cc_G(X)$), 
\textit{monotone} (i.e., $X\subseteq Y$ implies $\cc_G(X)\subseteq \cc_G(Y)$), and
\textit{idempotent} (i.e., $\cc_G(\cc_G(X))= \cc_G(X)$)
for all $X,Y \subseteq V(G)$.
We omit $G$ from $\cc_G$ if it is clear from the context.
For all graphs $G$ and $X \subseteq V(G)$, $\cc(X)$ can be computed
by iterating over \textit{all} elements $u \in \cc(X)$, starting with an arbitrary element of $X$, as follows  (see, e.g., \cite{Pelayo13}): 
Let $X' \supseteq X$ be the set of elements in $\cc(X)$ that have already been generated before we process the next element $u$. 
Then add $Y=\bigcup_{v \in X'} I(u,v)$ to $X'$, where $Y$ can be calculated by solving the \textit{single-source shortest path} (SSSP) problem (for unweighted graphs) from $u$ to all elements of $X'$. 
After all elements in $X'$ have been processed, we have $X'=\cc(X)$.
It is a folklore result that the SSSP problem can be solved with \textit{breadth-first search} (BFS) in $\bigO{n+m}$ time, implying that $\cc(X)$ can be computed in $\bigO{nm}$ time.

A graph $G$ is \textit{outerplanar}~\cite{Chartrand/Harary/1967} if it can be embedded in $\plane$ in a way that no two edges cross each other (except possibly in their endpoints) and there exists a point $P \in \plane$ such that each vertex of $G$ can be reached from $P$ by a simple curve that does not cross any of the edges. 
Removing all points and curves from the plane corresponding to the vertices and edges of $G$, respectively, we obtain a set of connected ``pieces'' of the plane, called \textit{faces}. 
Since $G$ is finite, all faces are bounded except for one, the \textit{outer} face. The bounded faces are called \textit{interior} faces.
The \textit{face number} of a biconnected outerplanar graph is the number of its interior faces; the face number of an outerplanar graph $G$, denoted  $\facenumber(G)$, is the maximum of the face numbers of its biconnected components.

Let $G$ be an outerplanar graph. All biconnected components, called \textit{blocks} of $G$ consist of a unique Hamiltonian cycle and a possibly empty set of (non-crossing) diagonals. 
Edges not belonging to blocks are called \textit{bridges}. 
The {\em block and bridge tree} (BB-tree) $\tG$ of $G$ is defined as follows~\cite{Horvath10}:
For each block $B$ of $G$, 
(i) introduce a new vertex, called \textit{block vertex} $v_{B}$,
(ii) remove all edges belonging to $B$, and
(iii) for every vertex $v$ of $B$, connect $v$ with $v_{B}$ by an edge if $v$ is adjacent to a bridge or to another biconnected component of $G$; otherwise remove $v$.
It holds that $\tG$ is a (free) tree and can be computed in $\bigO{n}$ time.

While $\cc(X)$ can be computed  in $\bigO{nm}$ time for arbitrary graphs, Thm.~\ref{thm:outerplanar_preclosure} gives rise to a faster algorithm for outerplanar graphs
\begin{theorem}[\cite{Allgeier09}] 
	\label{thm:outerplanar_preclosure}
	Let $G$ be an outerplanar graph. Then for all $X \subseteq V(G)$, $\cc(X) = \bigcup_{u,v \in X} I(u,v)$.
\end{theorem}
Thus, in case of outerplanar graphs, it suffices to perform a BFS only from the elements of $X$, resulting in the following corollary, by noting that $m = \bigO{n}$ in case of outerplanar graphs:
\begin{corollary}	
	\label{cor:outerplanar}
	Let $G$ and $X$ be as in Thm.~\ref{thm:outerplanar_preclosure}. Then $\cc(X)$ can be solved in time $\bigO{m|X|} = \bigO{n|X|}$.
\end{corollary}


\section{Spanning Outerplanar Subgraphs}
\label{sec:sampling}

The main contribution of this section is Alg.~\ref{alg:outerplanarSampling}, which generates a random spanning outerplanar subgraph for an undirected graph $G$ in time \textit{linear} in $m$. 
By Theorem~8 in \cite{Djidjev95}, for any graph $G$, one can generate a spanning \textit{planar} subgraph in $\bigO{m}$ time that is \textit{maximal} w.r.t. planarity. 
According to \cite{Djidjev95}, the algorithm for spanning planar subgraphs  can be modified in a way that it generates a \textit{maximal} spanning \textit{outerplanar} subgraph of $G$, also in $\bigO{m}$ time.
However, to the best of our knowledge, there exists no (simple) algorithmic realization of this result stated in \cite{Djidjev95} (cf. the discussion in Sect. 3.5 in \cite{Leip98}). 
We therefore propose an alternative algorithm that is \textit{easy} to implement and \textit{fast} in practice.
Our experimental results with Erd\H{o}s-R\'enyi random graphs clearly show that the spanning outerplanar graphs generated by our algorithm are \textit{almost maximal}.
More precisely, at most $0.3\%$ of the edges were missing for maximality. 
We note that Alg.~\ref{alg:outerplanarSampling} can easily be modified in a way that the face number of the output spanning outerplanar graphs becomes \textit{controllable} by some user specified upper bound.  
Thus, in the particular case that the face number is bounded by some \textit{constant}, the algorithm presented in \Cref{sec:closure} calculates the closure of any set of vertices in the output outerplanar graph in $\bigO{n}$ time.

Similarly to \cite{Fray2012}, our sampling algorithm is based on utilizing a basic property of \textit{Tr\'emaux trees}. 
More precisely, let $G$ be an undirected graph and $r \in V(G)$.
We assume w.l.o.g. that $G$ is connected.
Let $T$ be a DFS-tree of $G$ rooted at $r$. 
The connectivity of $G$ implies $V(T)=V(G)$.
We regard $T$ as a \textit{sorted} tree, where the order on the children of the vertices is defined by the DFS traversal of $G$. 
It is a well-known fact that $T$ is a Tr\'emaux tree of $G$, i.e., $v\preccurlyeq w$ or $w\preccurlyeq v$ holds for all \textit{back edges} $vw \in E(G) \setminus E(T)$, where for all $x,y \in V(T)$, $x \preccurlyeq y$ iff $\Path(r,y)$ in $T$ contains $x$.
In what follows, for all $v,w \in V(T)$, $p(v)$ denotes the parent of $v$ in $T$, $d(v)$ stands for the \textit{depth} of $v$ in $T$ (i.e., the length of $\Path(r,v)$ in $T$), and a back edge $vw \in E(G) \setminus E(T)$ with $w \preccurlyeq v$ is denoted by $(v,w)$.

It holds that the DFS traversal of $G$ defines a sequence of paths $P_1 = \Path(r_1,l_1),\ldots,P_k=\Path(r_k,l_k)$ in $T$, where $r_1 =r$, $l_1$ is the leftmost leaf of $T$, $r_{i+1}$ is the deepest vertex of $\Path(r,l_i)$ such that it has a child not belonging to $\bigcup_{\ell \leq i} V(P_\ell)$, and $l_{i+1}$ is the leftmost leaf in the subtree of $T$ rooted at $r_{i+1}$ that is not an element of $\bigcup_{\ell \leq i} V(P_\ell)$ $(1 \leq i < k)$.
For all back edges $(v,w)$ with $v\in V(P_i)$ we have that $w$ is a vertex of $\Path(r,l_i)$. 
The  sequence $P_1,\ldots,P_k$ is referred to as the (ordered) sequence of \textit{DFS paths} of $G$ w.r.t. $T$  (see Fig.~\ref{fig:outerplanar_example}a for an example with two paths).

Let $G'$ be a connected outerplanar graph, $T$ a DFS tree of $G'$ rooted at $r$ for some $r \in V(G')$, and let $\Path(r_1,l_1),\ldots,\Path(r_k,l_k)$ be the sequence of DFS paths of $G'$ w.r.t. $T$. 
Then $G'$ has an embedding in the plane with the following properties: For all leafs $l$ of $T$ and for all vertices $u$ on $P=\Path(r,l)$, the images of the vertices $v$ of $G'$ in $\mathbb{R}^2$ can be rotated around that of $r$, all in the same direction as $u$, such that the $x$-coordinate of the point representing $u$ becomes equal to that of $r$, the new embedding preserves the non-crossing edge property, and all back edges with both endpoints in $P$ are either on the left- or on the right-hand side of the vertical line containing $P$.
\begin{figure}
	\centering
	\begin{subfigure}[t]{0.45\linewidth}
		\centering
\begin{tikzpicture}[yscale=0.5]
\centering
\tikzmath{\xG1 = 0.0;  \xG2 = \xG1 + 2; \xd = \xG2 + 1.4;}

\newcommand{\DrawBox}[5]{
	\draw[#5] (#1, #2) -- (#1 + #3, #2) -- (#1 + #3, #2 - #4) -- (#1, #2 - #4) -- (#1, #2);
}

\node[inner sep=2pt] at (\xG1,1) {};
\node[circle, draw=black, inner sep=2pt] (R) at (\xG1,1) {};
\node[circle, draw=black, inner sep=2pt] (A) at (\xG1, 0) {};
\node[draw=black, circle, inner sep=2pt] (B) at (\xG1,-1) {};
\node[draw=black, circle, inner sep=2pt] (C) at (\xG1 - 0.5,-1) {};
\node[draw=black, circle, inner sep=2pt] (D) at (\xG1-0.5,-2) {};
\node[draw=black, circle, inner sep=2pt] (E) at (\xG1,-2) {};
\node[draw=black, circle, inner sep=2pt] (F) at (\xG1,-3) {};

\newcommand{\basegraph}{
\node[circle, draw=black, inner sep=2pt] (R) at (\xG1,1) {};
\node[circle, draw=black, inner sep=2pt] (A) at (\xG1, 0) {};
\node[draw=black, circle, inner sep=2pt] (B) at (\xG1,-1) {};
\node[draw=black, circle, inner sep=2pt] (C) at (\xG1-0.5,-1) {};
\node[draw=black, circle, inner sep=2pt] (D) at (\xG1-0.5,-2) {};
\node[draw=black, circle, inner sep=2pt] (E) at (\xG1,-2) {};
\node[draw=black, circle, inner sep=2pt] (F) at (\xG1,-3) {};

\draw[black, thick] (R) -- (A);
\draw[black, thick] (A) -- (B);
\draw[black, thick] (B) -- (C);
\draw[black, thick] (C) -- (D);
\draw[black, thick] (B) -- (E);
\draw[black, thick] (E) -- (F);

\draw[black, thick] (C) to [bend left = 30] (A);
\draw[black, thick] (D) to [bend left = 45] (R);
\draw[black, thick] (D) -- (B);

\path [black, thick] (F) edge [out=60,in=-60,bend angle=160,looseness=1](B);
\path [black, thick] (F) edge [out=200,in=-150,bend angle=200,looseness=1.2](A);
\path[black, thick] (F) edge [out=200,in=-160,bend angle=160,looseness=1] (R);
}

\basegraph

\begin{scope}[shift={(\xG2,0)}]
\node[inner sep=2pt, label=above:{\phantom{\textcolor{red}{$P$}}}] at (\xG1,1) {};
\node[circle, draw=black, inner sep=2pt, label=right:{$r=r_1$}] (R) at (\xG1,1) {};
\node[circle, draw=black, inner sep=2pt] (A) at (\xG1, 0) {};
\node[draw=black, circle, inner sep=2pt, label=right:$r_2$] (B) at (\xG1,-1) {};
\node[draw=black, circle, inner sep=2pt] (C) at (\xG1 - 0.5,-1) {};
\node[draw=black, circle, inner sep=2pt, label=below:$l_1$] (D) at (\xG1-0.5,-2) {};
\node[draw=black, circle, inner sep=2pt] (E) at (\xG1,-2) {};
\node[draw=black, circle, inner sep=2pt, label=right:$l_2$] (F) at (\xG1,-3) {};

\draw[red, thick, -latex] (R) -- (A);
\draw[red, thick, -latex] (A) -- (B);
\draw[red, thick, -latex] (B) -- (C);
\draw[red, thick, -latex] (C) -- (D);
\draw[red, thick, -latex] (B) -- (E);
\draw[red, thick, -latex] (E) -- (F);

\draw[black, dashed, thick, -latex] (C) to [bend left = 30] (A);
\draw[black, dashed,thick, -latex] (D) to [bend left = 45] (R);
\draw[black, dashed,thick, -latex] (D) -- (B);

\path [black, dashed,thick, -latex] (F) edge [out=60,in=-60,bend angle=160,looseness=1](B);
\path [black, dashed,thick, -latex] (F) edge [out=200,in=-150,bend angle=200,looseness=1.2](A);
\path[black, dashed, thick, -latex] (F) edge [out=200,in=-160,bend angle=160,looseness=1] (R);
\end{scope}

\end{tikzpicture}		
		\caption{\label{fig:outerplanar_dfs}\textit{Left:} Input graph $G$, \textit{Right:} DFS-Traversal of $G$ rooted at $r$ with DFS edges in red and back edges in dashed black. The paths are $P_1=\Path(r_1, l_1)$ and $P_2=\Path(r_2, l_2)$.}
		
	\end{subfigure}
	\hfill
	\begin{subfigure}[t]{0.45\linewidth}
		\centering
		\begin{tikzpicture}[yscale=0.5]
	\centering
	\tikzmath{\xG1 = 0.0; \xd = \xG1 + 1.5;}
	
	\newcommand{\DrawBox}[5]{
		\draw[#5] (#1, #2) -- (#1 + #3, #2) -- (#1 + #3, #2 - #4) -- (#1, #2 - #4) -- (#1, #2);
	}
	
	\node[inner sep=2pt, label=above:{\textcolor{red}{$P$}}] at (\xG1,1) {};
	\node[circle, draw=black, inner sep=2pt, label=right:{$r=r_1$}] (R) at (\xG1,1) {};
	\node[circle, draw=black, inner sep=2pt] (A) at (\xG1, 0) {};
	\node[draw=black, circle, inner sep=2pt, label=right:$r_2$] (B) at (\xG1,-1) {};
	\node[draw=black, circle, inner sep=2pt] (C) at (\xG1 - 0.5,-1) {};
	\node[draw=black, circle, inner sep=2pt, label=below:$l_1$] (D) at (\xG1-0.5,-2) {};
	\node[draw=black, circle, inner sep=2pt] (E) at (\xG1,-2) {};
	\node[fill, circle, inner sep=2pt, label=right:$l_2$] (F) at (\xG1,-3) {};

	\newcommand{\basegraph}{
		\node[circle, draw=black, inner sep=2pt] (R) at (\xG1,1) {};
		\node[circle, draw=black, inner sep=2pt] (A) at (\xG1, 0) {};
		\node[draw=black, circle, inner sep=2pt] (B) at (\xG1,-1) {};
		\node[draw=black, circle, inner sep=2pt] (C) at (\xG1-0.5,-1) {};
		\node[draw=black, circle, inner sep=2pt] (D) at (\xG1-0.5,-2) {};
		\node[draw=black, circle, inner sep=2pt] (E) at (\xG1,-2) {};
		\node[fill, circle, inner sep=2pt] (F) at (\xG1,-3) {};
		
		\draw[red, thick, -latex] (R) -- (A);
		\draw[red, thick, -latex] (A) -- (B);
		\draw[black, thick, -latex] (B) -- (C);
		\draw[black, thick, -latex] (C) -- (D);
		\draw[red, thick, -latex] (B) -- (E);
		\draw[red, thick, -latex] (E) -- (F);

		\draw[UniBlue, thick, -latex] (C) to [bend left = 30] (A);
		\draw[UniBlue, thick, -latex] (D) to [bend left = 45] (R);
	}
	\path [UniGray, dashed, thick, -latex] (F) edge [out=60,in=-60,bend angle=160,looseness=1](B);
	\path [UniGray, dashed, thick, -latex] (F) edge [out=200,in=-150,bend angle=200,looseness=1.2](A);
	\path[UniGray, dashed, thick, -latex] (F) edge [out=200,in=-160,bend angle=160,looseness=1] (R);
	\basegraph	
	\node[] (H1) at (\xd, 1.8) {$d$};
	\node[] (R1) at (\xd, 1) {$0$};
	\node[] (A1) at (\xd, 0) {$1$};
	\node[] (B1) at (\xd,-1) {$2$};
	\node[] (E1) at (\xd,-2) {$3$};
	\node[] (F1) at (\xd,-3) {$4$};
	
\end{tikzpicture}		
		\caption{\label{fig:add_edges_example}\textit{Blue Edges:} Already added left edges for path $\Path(r_1, l_1)$. \textit{Gray Edges:} Still unprocessed edges in the current path $P=\Path(r, l_2)$ drawn as a vertical line. $d$ denotes the node depth in $P$. The blue arrows denote valid back edges in $\Left$.}
		
	\end{subfigure}
	\\
	\begin{subfigure}[t]{\linewidth}
		\centering
\begin{tikzpicture}[yscale=0.5]
\centering
\tikzmath{\xLeft = 0.0; \xG1 = \xLeft + 0.0;  \xG2 = \xG1 + 3; \x1 = \xG2 + 4; \xL = \x1 + 1.8; \xR = \xL + 1.1; \xTL = \xR + 1.1; \xTR = \xTL + 1.1;}

\newcommand{\DrawBox}[5]{
	\draw[#5] (#1, #2) -- (#1 + #3, #2) -- (#1 + #3, #2 - #4) -- (#1, #2 - #4) -- (#1, #2);
}

\newcommand{\basegraph}{
	\node[inner sep=2pt, label=above:{\textcolor{red}{$P$}}] at (\xG1,1) {};
\node[circle, draw=black, inner sep=2pt, label=right:{$r=r_1$}] (R) at (\xG1,1) {};
\node[circle, draw=black, inner sep=2pt] (A) at (\xG1, 0) {};
\node[draw=black, circle, inner sep=2pt, label=right:$r_2$] (B) at (\xG1,-1) {};
\node[draw=black, circle, inner sep=2pt] (C) at (\xG1 - 0.5,-1) {};
\node[draw=black, circle, inner sep=2pt] (D) at (\xG1-0.5,-2) {};
\node[draw=black, circle, inner sep=2pt] (E) at (\xG1,-2) {};
\node[fill, circle, inner sep=2pt, label=right:$l_2$] (F) at (\xG1,-3) {};

\draw[red, thick, -latex] (R) -- (A);
\draw[red, thick, -latex] (A) -- (B);
\draw[black, thick, -latex] (B) -- (C);
\draw[black, thick, -latex] (C) -- (D);
\draw[red, thick, -latex] (B) -- (E);
\draw[red, thick, -latex] (E) -- (F);

\draw[UniBlue, thick, -latex] (C) to [bend left = 30] (A);
\draw[UniBlue, thick, -latex] (D) to [bend left = 45] (R);
}

\draw[thick] (\x1 -1, 1.5)--(\xTR + 1, 1.5);

\newcommand{\margin}{1}

\node[] (H1) at (\x1, 2) {$\reach$};
\node[] (R1) at (\x1, 1) {$\{L, R\}$};
\node[] (A1) at (\x1, 0) {$\{R\}$};
\node[] (B1) at (\x1,-1) {$\{L, R\}$};
\node[] (E1) at (\x1,-2) {$\{L, R\}$};
\node[] (F1) at (\x1,-3) {$\{L, R\}$};

\node[] (B12) at (\x1+\margin,-1) {\textcolor{orange}{$\rightarrow\{L\}$}};
\node[] (E12) at (\x1+\margin,-2) {\textcolor{orange}{$\rightarrow\{L\}$}};

\node[] (HL) at (\xL, 2) {$\bound_{L}$};
\node[] (RL) at (\xL, 1) {$F$};
\node[] (AL) at (\xL, 0) {$T$};
\node[] (BL) at (\xL,-1) {$T$};
\node[] (EL) at (\xL,-2) {$F$};
\node[] (FL) at (\xL,-3) {$F$};

\node[] (HR) at (\xR, 2) {$\bound_{R}$};
\node[] (RR) at (\xR, 1) {$F$};
\node[] (AR) at (\xR, 0) {$F$};
\node[] (BR) at (\xR,-1) {$F$};
\node[] (ER) at (\xR,-2) {$F$};
\node[] (FR) at (\xR,-3) {$F$};

\node[] (BR1) at (\xR+\margin/2,-1) {\textcolor{orange}{$\rightarrow T$}};
\node[] (ER1) at (\xR+\margin/2,-2) {\textcolor{orange}{$\rightarrow T$}};

\node[] (HTL) at (\xTL, 2) {$\lastL$};
\node[] (RTL) at (\xTL, 1) {$0$};
\node[] (ATL) at (\xTL, 0) {$0$};
\node[] (BTL) at (\xTL,-1) {$2$};
\node[] (ETL) at (\xTL,-2) {$2$};
\node[] (FTL) at (\xTL,-3) {$2$};

\node[] (FTL1) at (\xTL+\margin/2,-3) {\textcolor{orange}{$\rightarrow 3$}};

\node[] (HTR) at (\xTR, 2) {$\lastR$};
\node[] (TTR) at (\xTR, 1) {$0$};
\node[] (ATR) at (\xTR, 0) {$0$};
\node[] (BTR) at (\xTR,-1) {$1$};
\node[] (ETR) at (\xTR,-2) {$1$};
\node[] (FTR) at (\xTR,-3) {$1$};

\begin{scope}[shift={(\xG1,0)}]
\basegraph
\draw[UniBlue, dashed, thick, -latex] (F) to [bend left = 30] (B);
\path [UniGray, dashed, thick, -latex] (F) edge [out=150,in=-150,bend angle=200,looseness=1.2](A);
\path[UniGray, dashed, thick, -latex] (F) edge [out=160,in=-160,bend angle=160,looseness=1] (R);
\end{scope}

\begin{scope}[shift={(\xG2,0)}]
\basegraph
\path [orange, dashed, thick, -latex] (F) edge [out=50,in=-50,bend angle=160,looseness=1](A);
\path [orange, dashed, thick, -latex] (F) edge [out=60,in=-60,bend angle=160,looseness=1](B);
\path [UniGray, dashed, thick, -latex] (F) edge [out=40,in=-40,bend angle=10,looseness=0.8](R);
\end{scope}

\end{tikzpicture}
		\caption{\label{fig:outerplanar_algorithm}Running the subroutine \text{\sc AddEdges} for $l_2$. \textit{Left:} Determine $E_L$ (valid left back edges in dashed blue), \textit{Middle:} Determine $E_R$ (valid right back edges in dashed orange). \textit{Right:} The table shows the algorithm parameters before adding the orange edges (black numbers), changes after adding the orange edges to $\Right$ are marked in orange.}
	\end{subfigure}
	\caption{\label{fig:outerplanar_example}(\ref{fig:outerplanar_dfs}) shows an example of DFS-Traversal of a graph $G$, (\ref{fig:add_edges_example}) shows some possible intermediate result of Alg.~\ref{alg:outerplanarSampling} with unconsidered back edges starting at $l_2$ in dashed grey. (\ref{fig:outerplanar_algorithm}) shows the results of the \text{\sc AddEdges} subroutine with edges in $E_L$ marked in dashed blue and edges in $E_R$ marked in dashed orange.}
\end{figure}

Notice that if a back edge belongs to more than one path from a leaf to the root, then it is either to the left or to the right for \textit{all} such paths.
The set of left (resp. right) back edges of $G'$ is denoted by $\Left$ (resp. $\Right$).
A vertex $x$ of $G'$ lying on $\Path(r,l_i)$ is \textit{reachable} from \textit{left} (resp. \textit{right}) w.r.t. $\Path(r_i,l_i)$ if there is no left (resp. right) back edge $(v,w)$ with $w \prec x \prec v$ such that $x \prec r_i$ or 
$v \preccurlyeq l_i$.

Let $G$ be a connected graph, $T$ a DFS tree of $G$, and $G'$ be a spanning outerplanar subgraph of $G$ containing $T$ as a subgraph. 
We assume that $G'$ is embedded into the plane in the way sketched above, i.e., all back edges of $G'$ are either left or right back edges. 
A back edge $(v,w)\in E(G)\setminus E(G')$ is \textit{valid} if it can be added to $G'$ as a \textit{left} or \textit{right} back edge such that it intersects no other edges from $G'$ and for all vertices $v$ in the resulting graph there exists a path $P_i$ of $T$ such that $v$ is reachable from left or right w.r.t. $P_i$. 
One of the crucial steps in the generation of a spanning outerplanar graph of $G$ w.r.t. $T$ is to check the validity of back edges.  
To decide this problem in \textit{constant} time, we introduce some further notions.  More precisely, let $P_i=\Path(r_i,l_i)$ be a DFS path of $T$ and $v$ be a vertex with $r_i \preccurlyeq v$.
Then 
\begin{itemize}
	\item $\reach(v,P_i) \subseteq \{\dL, \dR\}$ denotes the direction(s) from which $v$ can be reached in $G'$ w.r.t. $P_i$,
	\item $\lastL(v)$ (resp. $\lastR(v)$) denotes the smallest depth of the vertex $w$ on $\Path(r,v)$ in $T$ such that $(v, w)$ is a valid left (resp. right) back edge, and
	\item $\bound_{L}(v)$ (resp. $\bound_{R}(v)$) is {\sc True} if there are $P=\Path(r, l_j)$ and $u,w \in V(P)$ for some $j$ $(1\leq j\leq k)$  such that $v\in V(P)$, $(u, w)$ is a left (resp. right) back edge, and $w\prec v \prec u$; o/w it is {\sc False}.
\end{itemize}

\begin{algorithm}[t]
	\KwIn{connected graph $G$}
	\KwOut{spanning outerplanar subgraph $H$ of $G$}
	select a vertex $r\in V(G)$ at random\label{line:selectrootvertex}\;
	generate a DFS tree $T$ of $G$ rooted at $r$ and with DFS paths $P_i=[v_1=r_i,\ldots,v_{n_i}=l_i], (1 \leq i \leq k)$\label{Line:1}\; 
	$\Left_0, \Right_0\leftarrow\emptyset$, $l \leftarrow 0$
	\label{line:update_init}\;
	$\lastL(r), \lastR(r) \leftarrow 0$\;
	$\bound_L(v), \bound_R(v) \leftarrow $ {\sc False} for all $v\in V(T)$\label{line:update_init_end}\;
	\For{$i=1,\ldots,k$\label{line:path_iteration}}{
		$\reach(r_i) \leftarrow \{\dL, \dR\}$\label{line:reachabiliy_r}\;
		\For{$\delta\in \{\dL,\dR\}$}{
			\textbf{if} {$\bound_{\delta}(r_i) \vee r_i=r$} \textbf{then}
			$\opn{{\last}_{\delta}}(r_i) = d(r_i)$\label{line:set_last_r}\;	 
			\textbf{else} ${\last}_\delta(r_i) = {\last}_\delta(p(r_i))$
			\label{line:set_last_r_parent}\;
		}
		\For{$j=2,\ldots,n_i$ \label{line:node_iteration}}{
			$\reach(v_j) \leftarrow \{\dL, \dR\}$\label{line:reachability_v}\;
			$\lastL(v_j) = \lastL(p(v_j))$,
			$\lastR(v_j) = \lastR(p(v_j))$\label{line:set_last_v}\;
			$F = \{(v_j, w)\in E :  w\prec v_j \}$\label{line:f_edges}\;
			$(E_L,E_R) = \text{\sc AddEdges}(v_j,F)$\label{line:add}\;
			$l \leftarrow l+1$\;
			$\Left_l = \Left_{l-1} \cup E_L, \Right_{l} = \Right_{l-1} \cup E_R$\label{line:path_iteration_end}\;
		}
	} 
	\Return $H=(V, E(T)\cup \Left_l \cup \Right_l)$\;
	\caption{\sc Spanning Outerplanar Subgraph}
	\label{alg:outerplanarSampling}
\end{algorithm}

Using the above notions and notation, we are ready to present Alg.~\ref{alg:outerplanarSampling} (see, also, Fig.~\ref{fig:outerplanar_example} for an example to demonstrate the algorithm).
In lines~\ref{line:selectrootvertex}--\ref{Line:1} it first computes a DFS tree of the input graph $G$ for some arbitrary root $r \in V(G)$ (see Fig.~\ref{fig:outerplanar_dfs}).
In lines~\ref{line:update_init}--\ref{line:update_init_end} it initializes some variables.
In particular, the left (resp. right) valid back edges that will be added to $T$ will be stored in the set variables $\Left_l$ (resp. $\Right_l$).
Since no back edge going out from the root can be added to $T$, $\lastL(r)$ and $\lastR(r)$ are both set to $0$.
Furthermore, the Boolean variables $\bound_L(v), \bound_R(v)$ are set to {\sc False} for all $v\in V$, as $T$ has no back edge initially.

The algorithm then processes the DFS paths $P_1,P_2,\ldots,P_k$ of $T$ in their DFS order defined above (cf. loop~\ref{line:path_iteration}--\ref{line:path_iteration_end}).
For each $P_i = \Path(r_i,l_i)$, it adds greedily as many as possible back edges to $\Path(r,l_i)$ with at least one endpoint in $P_i$ such that the extension does not violate outerplanarity.
In particular, it considers the vertices of $P_i$ one by one, from $r_i$ towards $l_i$ (Fig.~\ref{fig:add_edges_example} shows the new back edges added to $\Path(r_1, l_1)$ in blue). 
While processing the vertices of $P_i$, their reachability is set to $\{L, R\}$ (cf. lines~\ref{line:reachabiliy_r} and \ref{line:reachability_v}). Since we have not yet added any back edge to $P_i$, all of them are reachable from left as well as from right w.r.t. $P_i$. 
For simplicity, we omit the reference path $\Path(r,l_i)$ from the notation, by noting that all vertices above $r_i$ inherit their reachability state w.r.t. some previous path $\Path(r,l_j)$ with $j <i$.

For all vertices $v$ of $P_i$, $\lastL(v)$ (resp. $\lastR(v)$) is set to the depth of $r_i$ (cf. line~\ref{line:set_last_r}) if (i) $v = r_i$ and there is a $j$, $1 \leq j < i$, such that $r_i$ is \textit{not} reachable from left (resp. right) w.r.t. $\Path(r,l_j)$ or (ii) $v=r$; o/w it is set to $\lastL(p(v))$ (resp. $\lastR(p(v))$) (cf. lines~\ref{line:set_last_r_parent} and \ref{line:set_last_v}).  
Regarding the first case, there is no valid left/right back edge going out from $v$ satisfying (i) or (ii) and hence, its left/right smallest depth cannot be smaller than $d(v)$. 
For all other cases, if a back edge $(v, w)$ added to left (resp. right) destroys the reachability of some vertex $x$ with $w \preccurlyeq x$, then all other back edges $(v', w')$ with $v \preccurlyeq v'$ and $w'\preccurlyeq w$ added to left (resp. right) also destroy it. Hence, it suffices to store the depth of the vertex which hast the lowest depth and is a valid endpoint for a back edge added to left (resp. right).
After all relevant pieces of information have been calculated for $v_j$, we take the set of all possible back edges from $v_j$ ending in some vertex $w \prec v_j$ (line~\ref{line:f_edges}) and compute a maximal subset of this set of edges in function {\sc AddEdges} that can be added to $\Path(r,l_i)$ without destroying outerplanarity (line~\ref{line:add}) (see Fig.~\ref{fig:add_edges_example} and \ref{fig:outerplanar_algorithm}).

\begin{algorithm}[t]
	\nonl\textbf{Assumed:} undirected graph $G$ and DFS tree $T$ of $G$ \\ 
	\KwIn{$v \in V(G)$ and a set $F$ of back edges, all with initial vertex $v$}
	\KwOut{$E_L,E_R \subseteq F$ with $E_L = \emptyset$ or $E_R = \emptyset$}
	$E_L, E_R\leftarrow\emptyset$\;
	\For{$\delta\in \{L, R\}$\label{line:candidate_backedges_start}}{
		\For{$(v, w)\in F$}{
			\If{$\delta\in\reach(w)$ \textbf{and} ${\last}_\delta(v) \leq  d(w)$\label{line:condition}}{
				add $(v, w)$ to $E_\delta$\label{line:candidate_backedges_end}\;
			}
		}
	}
	$X\leftarrow L, Y\leftarrow R$\label{line:update_left}\;
	\textbf{if} $|E_R|>|E_L|$\label{line:update_right} \textbf{then}
	$X\leftarrow R, Y\leftarrow L$\;
	\If{$E_X \neq \emptyset$}{ 	
		${\last}_Y(v)$ = $d(p(v))$\label{line:enclosing2}\;
		\For{$(v, w)\in E_X$\label{line:update_prop}}{
			\For{$x$ in the open intervall $(v, w)$}{
				delete $X$ from $\reach(x)$\label{line:removing}\;
				${\last}_Y(x) = d(x)$\label{line:enclosing1}\;
				$\bound_X(x) = \text{\sc True}$\label{line:set_spanning}\;
			}	
		}
	}
	\textbf{if} $X = L$ \textbf{then} \Return $(E_L,\emptyset)$\label{line:XeqL}\;
	\textbf{else} \Return $(\emptyset,E_R)$\label{line:XeqR}\; 
	\caption{\sc Function AddEdges}
	\label{alg:check_valid_edges}
\end{algorithm}

Function {\sc AddEdges} is specified in Alg.~\ref{alg:check_valid_edges}.
Its input consists of a vertex $v$ of $T$ and a set $F$ of candidate back edges for $\Path(r,l_i)$ processed currently by Alg.~\ref{alg:outerplanarSampling}, each with starting vertex $v$.
Alg.~\ref{alg:check_valid_edges} tries to add as many as possible edges of $F$  to $\Path(r,l_i)$, either all from left or from right, without violating outerplanarity. 
In particular, each back edge $(v,w)$ is checked in loop~\ref{line:candidate_backedges_start}--\ref{line:candidate_backedges_end} for left and right validity w.r.t. $\Path(r,l_i)$ (cf. the condition in line~\ref{line:condition}) and, depending on the outcome of the test, is added to $E_L$ and $E_R$. As an example, all the gray edges in Fig.~\ref{fig:outerplanar_algorithm} violate at least one of the two conditions in line~\ref{line:condition} of Alg.~\ref{alg:check_valid_edges}, while the colored edges fulfill both of them.
Notice that once a back edge $(v, w) \in F$ has been added to one of the  sides of $\Path(r,l_i)$, then \textit{no} back edge $(v,w') \in F$ can be added to its other side, as $\reach(p(v),\Path(r,l_i))$ became empty, violating the reachability property of $p(v)$. 
Thus, we can add either all edges from $E_L$ to the left or all edges from $E_R$ to the right side of $\Path(r,l_i)$.
Since our goal is to maximize the number of back edges in $G'$, we take the set with the greater cardinality (cf. lines~\ref{line:update_left}--\ref{line:update_right}). 

After the selection of one of the two sets, say $E_L$ (the case of $E_R$ is analogous), we update the reachability information of the vertices on $\Path(r,v)$ as follows:
Since all back edges are of length at least $2$, no back edge $(s,t)$ with $l_i \succ s \succ v \succ p(v) \succ t$ can be added to the right of $\Path(r,l_i)$, as $p(v)$ became unreachable from both directions.
Therefore, $\lastR(v)$ has to be set to $d(p(v))$ (cf. line~\ref{line:enclosing2}). 
Furthermore, for all back edges $(v,w) \in E_L$ and for all internal vertices $x$ of $\Path(v,w)$, $x$ becomes unreachable from left (i.e., $L$ must be deleted from the reachability set of $x$ w.r.t. $\Path(r,l_i)$). 
Moreover, $(v,w)$ prohibits any left back edge in any possible path $P_j= \Path(r_j,l_j)$ with $x=r_j$ (i.e., $\bound_L(x)$ has to be set to $\text{\sc True}$) (cf. line~\ref{line:set_spanning}) and the terminal vertex of any right back edge in $P_j$ cannot be smaller w.r.t. the depth than $d(p(x))$ (i.e., $\lastR(x)$ has to be set to $d(x)$). 
In our example in Fig.~\ref{fig:outerplanar_algorithm}, by adding the orange edges to $\Right$ we update the parameters to the orange values (see, also, the table in Fig.~\ref{fig:outerplanar_algorithm}).

\begin{theorem} 
	\label{thm:outerplanar}
	For any connected graph $G$, Alg.~\ref{alg:outerplanarSampling} returns a spanning outerplanar subgraph of $G$ in $\bigO{m}$ time.
\end{theorem}
\begin{proof}
	The claim follows from the fact that the algorithm first generates a spanning subtree in Line~\ref{Line:1} and the results of Lemma~\ref{lem:correctness} and Lemma~\ref{lem:complexity} below.
\end{proof}

To prove Lemma~\ref{lem:correctness}, we use the following auxiliary lemma.
\begin{lemma}
	\label{lem:algorithm}
	For all $l \geq 0$, $\Left_l$ (resp. $\Right_l$) in Alg.~\ref{alg:outerplanarSampling} fulfills the properties:
	\begin{enumerate}[label = \arabic*)]
		\item[(i)] For all $(v_1, w_1), (v_2, w_2)$ in $\Left_l$ (resp. $\Right_l$) and $y \in V(G)$ satisfying $w_2\prec y \preccurlyeq v_1$ and $w_2 \prec y \preccurlyeq v_2$, it holds that
		\begin{eqnarray}
			w_1 \prec w_2 
			& \implies & v_2 \preccurlyeq v_1 \label{eq:implication1}\\
			w_1 = w_2 &\implies & v_1 \prec v_2 \text{ or } v_2 \prec v_1 \enspace . \label{eq:implication2}
		\end{eqnarray}
		\item[(ii)] For all $(v_a, w_a)$ in $\Left_l$ (resp. $\Right_l$) and $x \in V(T)$ with
		$w_a \prec x \prec v_a$, there is no $(v_b, w_b)\in \Right_l$ (resp. $\Left_l$) with $v_a\preccurlyeq v_b$ and $w_b\preccurlyeq w_a$.
	\end{enumerate}
\end{lemma}
\begin{proof}
	We prove (i) and (ii) by induction on $l$ for direction left; the proof for direction right is analogous.
	Furthermore, for (i) we show only (\ref{eq:implication1}); the proof of (\ref{eq:implication2}) is similar.
	The base case $l=0$ is trivial.
	For the induction step, let $(v_1,w_1),(v_2,w_2) \in \Left_{l+1}$.
	If $(v_1,w_1),(v_2,w_2) \in \Left_l$ ({\sc Case 1}), then (i) holds by the induction hypothesis.
	If $(v_1, w_1),(v_2, w_2)\in \Left_{l+1}\setminus\Left_l$ ({\sc Case 2}), then $v_1=v_2$ and $w_1\neq w_2$, implying (\ref{eq:implication1}). 
	If $(v_1, w_1)\in \Left_{l+1}\setminus\Left_l$ and $(v_2, w_2)\in \Left_l$ ({\sc Case 3}), then the order of processing the vertices of $T$ implies 
		$v_1\nprec v_2$. 
	Moreover, as $(v_1,w_1)$ is a left back edge, we have 
		$\lastL(v_1)  \leq  d(w_1)$ 
	(cf. the condition in line~\ref{line:condition} of Alg.~\ref{alg:check_valid_edges} for $\delta = L$).
	Suppose 
		$w_1 \prec w_2$. 
	Then $w_1\prec w_2 \prec v_1$. 
	Assume for contradiction that $v_1$ and $v_2$ are incomparable.
	Then they lie on different paths in $T$, implying $w_2 \prec y \prec v_2$ for $y$ in (i). 
	Since, by condition of this case, $(v_2,w_2)$ has been added to $T$ before $(v_1,w_1)$, $v_2$ was considered before $v_1$ in the DFS traversal. 
	Hence, there exists $r_i \prec v_1$ 
	such that  
	$w_2\prec r_i \prec v_2 \enspace$. 
	Thus, after $(v_2, w_2)$ has been added to $\Left_j$ for some $j \leq l$, we certainly have $L \notin\reach(r_i)$ and $\bound_{L}(r_i) = \text{\sc True}$ (cf. lines~\ref{line:removing} and \ref{line:set_spanning} in Alg.~\ref{alg:check_valid_edges}). 
	In a later step, when processing $v_1$, we therefore have
		$\lastL(v_1) \geq d(r_i)$ 
	(cf. lines~\ref{line:set_last_r} and \ref{line:set_last_r_parent} of Alg.~\ref{alg:outerplanarSampling}).  By $w_1 \prec w_2$ and $w_2\prec r_i \prec v_2$ 
	we have $w_1\prec w_2\prec r_i$, from which $d(w_1)<d(w_2)<d(r_i) \leq \lastL(v_1)$ follows by $\lastL(v_1) \geq d(r_i)$. 
	But this contradicts $\lastL(v_1)  \leq  d(w_1)$. 
	Thus, $v_1$ and $v_2$ are comparable and hence,  
	we have (\ref{eq:implication1}) by $v_1\nprec v_2$ for {\sc Case 3}. 
	The proof of (\ref{eq:implication1}) for $(v_1,w_1) \in \Left_l$ and $(v_2,w_2) \in \Left_{l+1}\setminus\Left_l$ ({\sc Case 4}) is analogous.	 
	
	Regarding claim~(ii), the base case 
	is trivial. 
	For the induction step, let $(v_a, w_a)\in \Left_{l+1}$ with $w_a \prec x \prec v_a$. 
	Assume first $(v_a, w_a)\in \Left_{l}$ and suppose for contradiction that there is a $(v_b,w_b) \in \Right_{l+1}$ with $v_a\preccurlyeq v_b$ and $w_b\preccurlyeq w_a$.
	The induction hypothesis implies $(v_b,w_b) \in \Right_{l+1}\setminus \Right_l$.  
	Then
	$
	\lastR(v_b)\geq \lastR(v_a)\geq d(p(v_a)) \geq d(x) > d(w_b) 
	$,
	where $\lastR(v_a)\geq d(p(v_a))$ holds by line~\ref{line:enclosing2} of Alg.~\ref{alg:check_valid_edges}. 
	Hence $(v_b,w_b)$ does \textit{not} satisfy the condition in line~\ref{line:condition} in Alg.~\ref{alg:check_valid_edges} for $\delta = R$, contradicting $(v_b,w_b) \in \Right_{l+1}$.
	A contradiction for the case that $(v_a, w_a)\in \Left_{l+1}\setminus\Left_l$ and $(v_b, w_b)\in \Right_l$ can be obtained in a similar way, by noting that we cannot have $(v_a, w_a)\in \Left_{l+1}\setminus\Left_l$ and 	$(v_b, w_b)\in \Right_{l+1}\setminus\Right_l$ (cf. lines~\ref{line:XeqL} and \ref{line:XeqR} of Alg.~\ref{alg:check_valid_edges}).	
\end{proof}

\begin{lemma}
	\label{lem:correctness}
	For all $l \geq 0$, $G_l = (V(G), E(T) \cup \Left_l \cup \Right_l)$ is outerplanar after iteration $l$ of loop~\ref{line:node_iteration}--\ref{line:path_iteration_end} of  Alg.~\ref{alg:outerplanarSampling}.
\end{lemma}
\begin{proof}
	We show by induction on $l$ that $G_l$ can be drawn in the plane in a way that
	\begin{itemize}
		\item[(i)] all edges $(v, w)\in \Left_l$ (resp. $(v, w)\in\Right_l$) that have been added to the DFS path $P_i$ for some $i$ $(1 \leq i \leq k)$ lie left (resp. right) w.r.t. $P_i$ and do not intersect any other edge of $G_l$ and
		\item[(ii)] all vertices of $G_l$ lie on the outer face.
	\end{itemize}
	
	The base case is trivial.
	For the induction step, suppose $G_{l}$ has an embedding in the plane satisfying (i)--(ii).
	If $\Leftnew_{l+1}=\Left_{l+1}\setminus\Left_l$ and $\Rightnew_{l+1}=\Right_{l+1}\setminus\Right_l$ are both empty then the claim holds by the induction hypothesis.
	Otherwise, exactly one of them, say $\Leftnew_{l+1}$, is non-empty by lines~\ref{line:XeqL}--\ref{line:XeqR} of Alg.~\ref{alg:check_valid_edges}; the proof of the case $\Rightnew_{l+1}\neq\emptyset$ is analogous.
	Let $P_i=\Path(r_i,l_i)$ be the DFS path $(1\leq i\leq k)$ and $v$ be a vertex of $P_i$ such that the left back edges in $\Leftnew_{l+1}$ have been constructed for $v$ in the outer loop of Alg.~\ref{alg:outerplanarSampling}.
	Then $v \neq r_i$ and it is the initial vertex of all back edges in $\Leftnew_{l+1}$.
	Thus, each edge in $\Leftnew_{l+1}$ can be drawn left w.r.t. $P_i$, without intersecting any other edge in $\Leftnew_{l+1}$. 
	%
	%
	All edges in $\Right_{l}$ with an endpoint in $\Path(r, l_i)$ are right w.r.t. $P_i$. Hence, the edges in $\Leftnew_{l+1}$ can be drawn without intersecting these right back edges.
	Suppose for contradiction that there is a new edge $(v_2, w_2)\in \Leftnew_{l+1}$ that cannot be drawn left w.r.t. $P_i$ without crossing some other edge $(v_1,w_1) \in \Left_l$.
	By the induction hypothesis, $(v_1,w_1)$ could be drawn for some $l' \leq l$ iteration also left w.r.t. $P_i$, without crossing any other edge. 
	Hence,  we must have $w_1\prec w_2$, $w_2 \prec y \preccurlyeq v_1$, and $w_2 \prec y \preccurlyeq v_2$, where $y$ is the vertex with the largest depth satisfying $y \preccurlyeq v_1,v_2$.
	But then, $v_2\preccurlyeq v_1$ by \Cref{lem:algorithm}, contradicting that $(v_1, w_1)$ has been considered before $(v_2,w_2)$. 
	Thus, $(v_2, w_2)$ can be drawn left w.r.t. $P_i$ without intersecting any edges in $\Left_{l}$ and hence, (i) holds.
	
	To prove (ii), notice that if $(v, w)\in \Left_{l+1}$ destroys  outerplanarity, then $(v, w') \in \Left_{l+1}$ with $w'\prec w$ does the same. 
	Thus, it suffices to consider the back edge in $\Left_{l+1}$ with the terminal vertex of the smallest depth. Let $(v,w^*)$ be this back edge.
	We show that it is possible to add $(v, w^*)$ to the planar embedding such that all vertices $x$ lie on the outer face and (i) stays valid.
	This is straightforward by induction for all vertices $x \in V(T)$ with $w^* \nprec x$, so it suffices to consider $V_1= V(\Path(w^*,v))$ and $V_2 = \{x\in V(T)\setminus V_1: w^*\prec x\}$.  
	
	We first prove the claim for the vertices in $V_1$.
	Suppose for contradiction that there is a vertex $x\in V_1$ that does not lie on the outer face.
	This can happen iff there is $(v_R,w_R) \in \Right_l$ such that $w^*\prec x \prec v$ and $w_R\prec x \prec v_R$. 
	But this contradicts (ii) of \Cref{lem:algorithm}. 
	Regarding the other case, assume there is an $x \in V_2$ that does not lie on the outer face. 
	Let $x^* \in V_1$ be the vertex with maximum depth such that $x^*\prec x,v$. Assume there is an edge $(v', w')\in \Left_{l}$ with $w'\prec x^*\prec v'$. 
	Since one of $w'\prec w^*$, $w^*\prec w'$, and $w'=w^*$ holds, the condition of (i) of \Cref{lem:algorithm} is fulfilled for $x^*=y$. 
	But $v',v$ are incomparable, implying that such an edge does not exists. 
	Hence, we can redraw the outerplanar subgraph consisting of all vertices $y$ with $y\succcurlyeq x^*$ right to $P_i$ such that all of its vertices and all other vertices of the graph lie on the outer face. Moreover, it can be redrawn such that no edges are crossing and this new embedding 
	fulfills (i) because no back edge in the subgraph of all vertices $y$ with $y\succcurlyeq x^*$ lies left w.r.t. $P_i$. This completes the proof of (ii).
\end{proof}

\begin{lemma}
	\label{lem:complexity}
	Alg.~\ref{alg:outerplanarSampling} terminates in $\bigO{|E(G)|}$ time.
\end{lemma}
\begin{proof}
	Note that $T$ in line~\ref{Line:1} can be computed in $\bigO{|E(G)|}$ time and
	%
	%
	{\sc AddEdges} is called at most $|V(G)|$ times in line~\ref{line:add}.
	It can be checked in constant time whether an edge in $F$ can be added to $E_L$ or $E_R$. 
	If a back edge $(v, w)$ can be added, we have to update the properties of vertices between $(v, w)$ (cf. line~\ref{line:update_prop}). 
	This can be done with a \textit{na\"{\i}ve} algorithm in quadratic time. 
	However, we can store the vertices during the iteration over all $P_i$ in a \textit{global} stack that are reachable from both \textit{left} and \textit{right}. 
	If a new edge $(v, w)$ is added, we remove all vertices starting with the parent of $v$ from the stack, unless we have found $w$ ($w$ will not be deleted from the stack); this is because these vertices cannot be the endpoints of other \textit{left} or \textit{right} edges. The runtime of this operation is \textit{linear} in the number of elements removed from the stack. 
	Once a vertex has been removed from the stack, it will never be added again to it, except for the case that it is equal to $r_i$ for some $i$ $(1\leq i \leq k)$. Hence the \textit{overall} runtime of this operation is at most linear in the number of edges, implying the claimed total runtime of $\bigO{|E(G)|}$.
\end{proof}


\section{Closures in Outerplanar Graphs}
\label{sec:closure}
This section deals with the following problem for \textit{outerplanar} graphs: 
\begin{problem}
\label{problem:batch}
	\textit{Given} a graph $G=(V,E)$ and $X \subseteq V$, \textit{compute} $\cc(X)$. 
\end{problem}
As discussed in Section~\ref{sec:preliminaries}, Problem~\ref{problem:batch} can be solved in $\bigO{nm}$ time for arbitrary and in $\bigO{n|X|}$ time for outerplanar graphs.
The complexity of the algorithm presented in this section for outerplanar graphs is $\bigO{nf}$, where $f$ is the face number of $G$. 
Thus, its complexity is \textit{independent} of the cardinality of the input set $X$. 
Since $f = \bigO{n}$, it does not improve the \textit{theoretical} worst-case complexity $\bigO{n|X|}$.
It has, however, two advantages over the algorithm sketched in Sect.~\ref{sec:preliminaries}. The first one is \textit{practical}: Our experiments with various graphs clearly show that the face number of spanning outerplanar graphs is \textit{negligible}, compared to their size (i.e., $n$).  
The second one is of \textit{theoretical} interest: Allowing only at most $c$ faces per biconnected components in the spanning outerplanar graphs for some \textit{constant} $c$, our algorithm runs in guaranteed \textit{linear} time. 

\begin{algorithm}[t]
	\KwIn{outerplanar graph $G$ and $X \subseteq V(G)$}
	\KwOut{$\cc(X)$}

	construct the BB-tree $\tG$ for $G$\label{alg:BBTree}\;
	$X_0 \gets X$ \label{alg:initX}\;
	$Y \gets \text{ set of block nodes of $\tG$}$\label{alg:initY}\;
	$C_1 = \{v_B\in Y:  V(B) \cap X_0 \neq \emptyset\}$\label{alg:C1}\;
	$C_2 \gets V(\tG) \cap X_0$ \label{alg:C2}\;
	$C \gets \closuretree(\tG,C_1 \cup C_2)$, \label{alg:setC}\;
	$X_1 \gets X_0 \cup (C \cap V(G))$,  $i \gets 1$\label{alg:initX1}\;
	\ForEach{$v_B \in Y \cap C$ \label{alg:startfor}}{
		\If{$|V(B) \cap X_i| > 1$}{
			$X_{i+1} \gets X_i \cup \beta(B,V(B) \cap X_i)$\label{alg:updateX}\;
			$i \gets i+1$\;
		}
	}
	\label{alg:endfor}
	\Return $X_i$\label{alg:outputX}\;
	\caption{\sc Outerplanar Graphs: Closure }
	\label{alg:batchclosure}
\end{algorithm}

The algorithm solving Problem~\ref{problem:batch} for outerplanar graphs is given in Alg.~\ref{alg:batchclosure}.
We assume that $G$ is \textit{connected}, by noting that all results can easily  be generalized to disconnected outerplanar graphs as well.
Alg.~\ref{alg:batchclosure} first calculates the BB-tree $\tilde{G}$ for the input outerplanar graph $G$ and then stores $X$ and the set of block nodes of $\tilde{G}$ in the variables $X_0$ and $Y$, respectively (lines~\ref{alg:BBTree}--\ref{alg:initY}).
In line~\ref{alg:C1}, it computes the set $C_1$ of block nodes representing such blocks of $G$ that have at least one vertex from $X_0$.
In a similar way, $C_2$ contains the set of nodes of $\tilde{G}$ that belong to $X_0$ (cf. line~\ref{alg:C2}).
The closure of $C_1\cup C_2$ in $\tilde{G}$ is calculated in $C$ (line~\ref{alg:setC}) and the union of $X_0$ and the set of vertices in $C$ that belong to $V(G)$ is stored in $X_1$ (line~\ref{alg:initX1}).
Note that at this point of the algorithm we have $v \in X_1 \subseteq \ccG(X)$ for all $v \in \ccG(X)$ not belonging to a biconnected component of $G$. 
Furthermore, for all $v \in \ccG(X) \setminus X_1$, $v$ is on a shortest path in one of the blocks and with both endpoints in $X$.  
Accordingly, 
in loop~\ref{alg:startfor}--\ref{alg:endfor}, the algorithm takes all block nodes $v_B$ of $\tilde{G}$ that belong to the closed set $C$, computes the closure of the set of vertices of the corresponding block $B$ over $B$ that are known to be closed (i.e., belong to $X_i$), updates the set of already known closed vertices in $X_{i+1}$, and increments the loop variable $i$.
At the end, it returns the set $X_i$.

\begin{algorithm}[tb]
	\KwIn{tree $T$ and $X \subseteq V(T)$}
	\KwOut{$\cc_T(X)$}

	\While{$\exists v \in V(T) \setminus X \text{ with } \degree{v} \leq 1$}{
		remove $v$ from $T$\;
	}
	\Return $V(T)$\;

	\caption{\sc Function $\closuretree$}
	\label{alg:batchtree}
\end{algorithm}

It remains to discuss functions $\closuretree$ and $\closureblock$ (cf. lines~\ref{alg:setC} and \ref{alg:updateX}).
Regarding $\closuretree$ (see Alg.~\ref{alg:batchtree}), it computes the closure of a set of nodes of a tree.
It iteratively removes all leaves of $T$ that are not in $X$ and returns the set of all nodes of $T$ at the end that have not been deleted. 
The proof of the following lemma is straightforward:
\begin{lemma}
\label{lm:correctnesstreeclosure}
For any tree $T$ with $n$ nodes and for any $X \subseteq V(T)$, Alg.~\ref{alg:batchtree} returns $\cc_T(X)$ in $\bigO{n}$ time.
\end{lemma} 

\begin{algorithm}[t]

	\KwIn{biconnected outerplanar graph $B$, $X \subseteq V(B)$}
	\KwOut{$G_X \subseteq X$ such that $\cc_B(G_X) =\cc_B(X)$}

	$G_X \gets \emptyset$ 	\hfill // $G_X \subseteq X$: generator set for $\cc_B(X)$ \label{alg:startG}\;
	\ForAll{interior faces $F$ of $B$}{
		$X' \gets V(F) \cap X$\; 
		\If{$|X'| > 0$}{
			select an arbitrary vertex $w$ from $X'$\;
			add $u = \argmax\limits_{x \in X'} d(x,w)$ to $G_X$ \label{alg:u}\;
			add $v = \argmax\limits_{x \in (X' \setminus \cc_F(\{u,w\})) \cup \{w\} } d(x,w)$  to $G_X$ \label{alg:v}\;
			\If{$w \notin \cc_F(\{u,v\})$ \label{alg:uvw}}{
				add $w$ to $G_X$ \label{alg:endG}\;
			}
		}		
	}
	\Return $G_X$\;
\caption{\sc Function \generatorset}
\label{alg:generator}
\end{algorithm}
Regarding $\closureblock$ (see Alg.~\ref{alg:batchblock}), which computes the closure over biconnected outerplanar graphs, we first show that for any biconnected outerplanar graph $B$ with $f=\facenumber(B)$ and for any $X \subseteq V(B)$, there is a set $G_X \subseteq X$ of cardinality \textit{linear} in $f$ such that $\cc_B(G_X) = \cc_B(X)$.
Furthermore, $G_X$ can be constructed in linear time as follows (see, also, Alg.~\ref{alg:generator}): 
Initialize $G_X$ with $\emptyset$ (cf. line~\ref{alg:startG}) and process all interior faces $F$ of $B$ one by one in an arbitrary order as follows: 
If $F$ has no vertex from $X$ then disregard $F$; o/w choose an arbitrary vertex $w$ from $X' = V(F) \cap X$. 
For that $w$, calculate the furthest vertex $u \in X'$ and the furthest vertex $v \in (X' \setminus \cc_F(\{u,w\})) \cup \{w\}$, and add $u$ and $v$ to $G_X$ (cf. lines~\ref{alg:u} and \ref{alg:v} of Alg.~\ref{alg:generator}).  
Note that $\cc_F(\{u,w\})) = V(F)$ if $d(u,w) = \ell/2$, where $\ell$ is the (cycle) length of $F$; o/w it is the set of vertices of the (unique) shortest path between $u$ and $w$.
If $w$ does not lie on a shortest path between $u$ and $v$ (cf. line~\ref{alg:uvw}), then add $w$ to $G_X$ as well.
Note that $u$ and $v$ can be equal to $w$. 
Hence, we add at least one and at most three vertices of $X'$ to $G_X$ for $F$. As an example, consider the biconnected outerplanar graph $B$ and the set $X \subseteq V(B)$ marked with color blue in Fig.~\ref{fig:biconnected_example}. A generator set $G_X$ containing the four vertices (red) is given in the middle. In case of the largest face of $G$, suppose we first select $w \in X$. For $w$, we first add $u$ and then $v$ to $G_X$ by the algorithm; $w$ is not added because it is on a shortest path between $u$ and $v$. The closure $\cc(X)=\cc(G_X)$ is given on the right-hand side of Fig.~\ref{fig:biconnected_example}. 
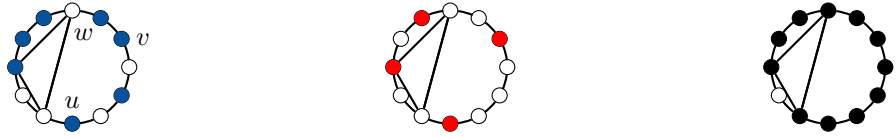
\begin{figure}[t]
	\centering
	\hfill
	\begin{subfigure}[t]{0.25\linewidth}
\begin{tikzpicture}
\centering
\newcommand{\basegraph}{
	\foreach \a in {1,2,...,12} \draw[fill=white, black, circle, inner sep=1pt] (\a*360/12:0.75cm) node[draw=black, circle, inner sep=2pt](\a){};
	\foreach \a in {1,2,...,12} \draw[thick] (\a*360/12+7:0.75cm) arc (\a*360/12+7:\a*360/12+22:0.75cm);
	
	\draw[thick] (3)--(6);
	\draw[thick] (8)--(6);
	\draw[thick] (3)--(8);
	\draw[thick] (3)--(8);
}

\newcommand{\inputset}{
	\foreach \a in {1, 2, 4, 11, 9, 6, 5} \draw[fill, UniBlue, circle, inner sep=1pt] (\a*360/12: 0.75cm) node[fill, circle, inner sep=2pt](\a){};
}
\draw[fill, black, circle, inner sep=1pt] (2*360/12: 0.75cm) node[fill, circle, inner sep=2pt, label=below left:{$w$}](X){};
\draw[fill, black, circle, inner sep=1pt] (1*360/12: 0.75cm) node[fill, circle, inner sep=2pt, label=right:{$v$}](X){};
\draw[fill, black, circle, inner sep=1pt] (9*360/12: 0.75cm) node[fill, circle, inner sep=2pt, label=above:{$u$}](X){};
\basegraph
\inputset
\end{tikzpicture}
	\end{subfigure}
	\hfill
	\begin{subfigure}[t]{0.25\linewidth}
\begin{tikzpicture}
\centering
\newcommand{\basegraph}{
	\foreach \a in {1,2,...,12} \draw[fill=white, black, circle, inner sep=1pt] (\a*360/12:0.75cm) node[draw=black, circle, inner sep=2pt](\a){};
	\foreach \a in {1,2,...,12} \draw[thick] (\a*360/12+7:0.75cm) arc (\a*360/12+7:\a*360/12+22:0.75cm);
	
	\draw[thick] (3)--(6);
	\draw[thick] (8)--(6);
	\draw[thick] (3)--(8);
	\draw[thick] (3)--(8);
}

\newcommand{\genset}{
	\foreach \a in {1, 4, 9, 6} \draw[fill, red, circle, inner sep=1pt] (\a*360/12:0.75cm) node[fill, circle, inner sep=2pt](\a){};
}
\basegraph
\genset
\end{tikzpicture}
	\end{subfigure}
	\hfill
	\begin{subfigure}[t]{0.25\linewidth}
\begin{tikzpicture}
\centering
\newcommand{\basegraph}{
	\foreach \a in {1,2,...,12} \draw[fill=white, black, circle, inner sep=1pt] (\a*360/12:0.75cm) node[draw=black, circle, inner sep=2pt](\a){};
	\foreach \a in {1,2,...,12} \draw[thick] (\a*360/12+7:0.75cm) arc (\a*360/12+7:\a*360/12+22:0.75cm);
	
	\draw[thick] (3)--(6);
	\draw[thick] (8)--(6);
	\draw[thick] (3)--(8);
	\draw[thick] (3)--(8);
}

\newcommand{\closedset}{
	\foreach \a in {1,2,3,4,5,6,8,9,10,11,12} \draw[fill=red, black, circle, inner sep=1pt] (\a*360/12: 0.75cm) node[fill, black, circle, inner sep=2pt](\a){};
}
\basegraph
\closedset
\end{tikzpicture}
	\end{subfigure}
	\hfill
	\caption{\label{fig:biconnected_example}(left) Biconnected outerplanar graph $B$ with $X\subseteq V(B)$ in blue, (middle) generator set $G_X\subseteq X$ in red, (right) $\cc(X)=\cc(G_X)$.}
\end{figure}
We have the following result about Alg.~\ref{alg:generator}:
\begin{lemma}
\label{lm:generatorset}
Let $B$ be a biconnected outerplanar graph with $f=\facenumber(B)$. Then for all $X \subseteq V(B)$, Alg.~\ref{alg:generator} computes a set $G_X \subseteq X$ in $\bigO{n}$ time such that $\cc_B(G_X) = \cc_B(X)$ and $|G_X| = \bigO{f}$.
\end{lemma} 

\begin{proof}
	Since $G_X \subseteq X$, $\cc_B(G_X) \subseteq \cc_B(X)$ follows from the monotonicity of $\cc_B$. 
	We show $\cc_B(X) \subseteq \cc_B(G_X)$ by induction on $f$.
	The base case $f =1$ is trivial if $B$ has at most two vertices from $X$. Otherwise, let $u,v,w$ be the vertices considered by Alg.~\ref{alg:generator} for $F=B$. 
	Since $|V(B) \cap X| \geq 2$, we have $u \neq w$.
	If $v =w$, then $G_X = \{u,v\}$ and $X \subseteq \cc_B(G_X)$, from which the monotonicity and idempotency of $\cc_B$ imply $\cc_B(X) \subseteq \cc_B(G_X)$.
	If $v \neq w$, then $u,v$, and $w$ are pairwise different. Furthermore, by definition of this case, $w$ does not lie on the (unique) shortest path between $u$ and $v$. But then 
	$
	X \subseteq V(B) = \bigcup_{x,y \in G_X} \cc_B(\{x,y\})) = \cc_B(G_X) 
	$,
	where the last equality holds by Thm.~\ref{thm:outerplanar_preclosure}. from which $\cc_B(X) \subseteq \cc_B(G_X)$ follows, again by monotonicity and idempotency.
	For the induction step, 
	let $B$ be a biconnected outerplanar graph with interior faces $F_1,\ldots,F_{f+1}$ for some $f \geq 1$. 
	We can assume w.l.o.g. that $F= F_{f+1}$ is adjacent to exactly one interior face. 
	Then $F_1,\ldots,F_f$ form a biconnected outerplanar graph $B'$. 
	Let $X_1 = X \cap V(B')$ (resp. $X_2 = X \cap V(F)$) and $G_{X_1}$ (resp. $G_{X_2}$) be the generator set constructed for $B'$ (resp. $F$) by Alg.~\ref{alg:generator}. 
	Note that $G_X = G_{X_1} \cup G_{X_2}$.
	Thm.~\ref{thm:outerplanar_preclosure} implies  
	\begin{equation}
		\cc_B(X) = \cc_B(X_1) \cup \cc_B(X_2) \cup \bigcup_{u \in X_1, v \in X_2} \cc_B(\{u,v\}) \label{eq:gensetA} 
		\enspace . 	
	\end{equation}
	We have
	\begin{alignat}{4}
		&\cc_B(X_1) &&\subseteq \cc_{B'}(G_{X_1}) &&= \cc_{B}(G_{X_1}) &&\subseteq \cc_B(G_X) \label{eq:gensetB1} \\
		&\cc_B(X_2) &&\subseteq \cc_{F}(G_{X_2})  &&= \cc_{B}(G_{X_2}) &&\subseteq \cc_B(G_X)   \label{eq:gensetB2}
	\end{alignat}
	by $\cc_B(X_1) = \cc_{B'}(X_1)$, $\cc_{B'}(G_{X_1}) = \cc_{B}(G_{X_1})$ and $\cc_B(X_2) = \cc_{F}(X_2)$, $\cc_F(G_{X_2}) = \cc_{B}(G_{X_2})$, and by the induction hypothesis to $B'$ and $F$. 
	Below we show that for all $u \in X_1, v \in X_2$,
	\begin{equation}
		\cc_B(\{u,v\}) \subseteq \bigcup\limits_{x,y \in G_{X}} \cc_B(\{x,y\}) = \cc_B(G_X)\enspace , \label{eq:gensetC}
	\end{equation}
	from which $\cc_B(X) \subseteq \cc_B(G_X)$ follows by (\ref{eq:gensetA})--(\ref{eq:gensetB2}).
	To prove (\ref{eq:gensetC}), let $u \in X_1$, $v \in X_2$, and let $F_u$ be the interior face of $B$ containing $u$. 
	By construction, there are $u_1,u_2 \in G_{X_1}$ such that $u \in \cc_{B'}(\{u_1,u_2\})$ and $\cc_{B'}(\{u_1,u_2\}) \cap G_{X_1} = \{u_1,u_2\}$.
	Similarly, there are $v_1,v_2 \in G_{X_2}$ with $v \in \cc_{F}(\{v_1,v_2\})$ and $\cc_{F}(\{v_1,v_2\}) \cap G_{X_2} = \{v_1,v_2\}$.
	It holds that for all shortest paths $P_{u,v}$ between $u$ and $v$, there is a shortest path $P_{u',u} \oplus P_{u,v} \oplus P_{v,v'}$ that contains $P_{u,v}$, where $\oplus$ denotes the path concatenation operation and $P_{u',u}$ (resp. $P_{v,v'}$) is a shortest path from some $u' \in \{u,u_1,u_2\}$ to $u$ (resp. from $v$ to some $v' \in \{v,v_1,v_2\}$).
	One can easily check that $V(P_{u',u}) \subseteq \cc_B(\{u',u\})$ and $V(P_{v,v'}) \subseteq\cc_B(\{v,v'\})$ are both subsets of $\cc_{B}(G_X)$, from which (\ref{eq:gensetC}) holds by $V(P_{u',v'})\subseteq\cc_B(G_X)$, completing the proof of $\cc_B(X) \subseteq \cc_B(G_X)$.
	
	The linear time complexity of Alg.~\ref{alg:generator} follows from the facts that 
	each iteration of the loop can be carried out in $\bigO{|V(F)|}$ time and 
	the sum of the sizes of the faces $F$ is $\bigO{n}$.
\end{proof}


\begin{algorithm}[tb]
	\KwIn{biconnected outerplanar graph $B$, $X \subseteq V(B)$}
	\KwOut{$\cc_B(X)$}
	
	$G_X \gets \generatorset(B,X)$\;
	\Return $\cc_B(G_X)$\;
	
	\caption{\sc Function $\beta$}
	\label{alg:batchblock}
\end{algorithm}

We are ready to present Alg.~\ref{alg:batchblock} computing the closure of a set of vertices over a biconnected outerplanar graph (see line~\ref{alg:updateX} in Alg.~\ref{alg:batchclosure}).
The input of Alg.~\ref{alg:batchblock} consists of a biconnected outerplanar graph $B$ and a set $X \subseteq V(B)$. 
Using Alg.~\ref{alg:generator}, it first computes a generator set $G_X$ for $B$ and $X$  and computes then $\cc_B(X)=\cc_B(G_X)$ in time $\bigO{|V(B)|\cdot |G_X|}$ by Cor.~\ref{cor:outerplanar}. 
\begin{lemma}
\label{lm:correctnessblockclosure}
Let $B$, $f$, and $X$ be as in Lemma~\ref{lm:generatorset}.
Then Alg.~\ref{alg:batchblock} computes $\cc_B(X)$ correctly and in $\bigO{|V(B)|f}$ time. 	
\end{lemma} 

In order to state Thm.~\ref{th:batch}, the main result of this section, we need some further notation.
For any $v \in V(\tilde{G})$, $\Gamma(v)$ denotes $\{v\}$ if $v \in V(G)$; o/w $\Gamma(v) = V(B)$, where $B$ is the block of $G$ represented by $v$. 
Prop.~\ref{pr:BBtreeshortestpath} below is used in the proof of Thm.~\ref{th:batch}. Its proof follows from the definitions.
\begin{proposition}  
	\label{pr:BBtreeshortestpath}
	Let $G$ be an outerplanar graph and $x \in V(G)$.
	\begin{itemize}
		\item[(i)] Let $\tilde{G}$ be the BB-tree of $G$ and $u,v \in V(\tilde{G})$. 
		If $x\in V(\tilde{G})$, then $x$ is on the shortest path in $\tilde{G}$ between $u$ and $v$ iff it is on a shortest path in $G$ between $u'$ and $v'$, for all $u' \in \Gamma(u)$ and $v' \in \Gamma(v)$.
		\item[(ii)] Let $B$ be some block of $G$ and $u,v \in V(B)$. 
		Then $x$ is on a shortest path in $B$ connecting $u$ and $v$ iff it is on a shortest path in $G$ between $u$ and $v$. 
	\end{itemize}
\end{proposition}
	
\begin{theorem}
	\label{th:batch}
	For outerplanar graphs, Alg.~\ref{alg:batchclosure} solves Problem~\ref{problem:batch} correctly and in $\bigO{n f}$ time, where $f=\facenumber(G)$.
\end{theorem}

\begin{proof}
	Regarding the correctness, for the $X_i$s in Alg.~\ref{alg:batchclosure} we have  $X=X_0 \subseteq X_1 \subseteq \ldots \subseteq X_N \subseteq V(G)$  (cf. lines~\ref{alg:initX}, \ref{alg:initX1}, and \ref{alg:updateX}), where $N$ is the value of $i$ at termination.
	We show that $X_N = \cc(X)$.
	This is straightforward for $|X| \leq 1$, so assume $|X| > 1$.
	We prove the soundness (i.e., $X_N \subseteq \cc(X)$) by showing with induction on $i$ that $X_i \subseteq \cc(X)$ 
	for all $i$. 
	The proof of the case $i=0$ is automatic by the extensivity of $\cc$.
	For $i=1$, the same argument holds if $x \in X_0$, so consider the case that $x \in X_1 \setminus X_0$. 
	Then, by Lemma~\ref{lm:correctnesstreeclosure} concerning the correctness of  $\closuretree$ computing the closure over trees (cf. line~\ref{alg:setC}), $x$ belongs to the closure of $C_1 \cup C_2$ in $\tilde{G}$. 
	That is, there are $u,v \in C_1 \cup C_2 \subseteq V(\tilde{G})$ such that $x$ lies on a shortest path connecting $u$ and $v$ in $\tilde{G}$, from which we have $x \in \cc(X)$ by (i) of Prop.~\ref{pr:BBtreeshortestpath}.
	For the induction step, suppose $X_k \subseteq \cc(X)$ holds for $k \geq 1$ and let $x \in X_{k+1}$.
	If $x \in X_k$, then $x \in \cc(X)$ by the induction hypothesis.
	Otherwise, by the definition of $k$, $x$ has been added to $X_{k+1}$ in line~\ref{alg:updateX}. 
	But then, $x \in \cc(X)$ is immediate from (ii) of Prop.~\ref{pr:BBtreeshortestpath} by Lemma~\ref{lm:correctnessblockclosure} concerning the correctness of  $\closureblock$ computing the closure over biconnected outerplanar graphs, and by the induction hypothesis, completing the proof of soundness.   	
	
	For the completeness (i.e., $\cc(X) \subseteq X_N$), let $x \in \cc(X)$.
	Clearly, $x \in X_N$ if $x \in X$.
	Otherwise, by Thm.~\ref{thm:outerplanar_preclosure}, there are $u,v \in X$ with $u \neq v$ such that $x \in \cc(\{u,v\})$. 
	If $x$ does not belong to a block in $G$, then $x \in V(\tilde{G})$ and $\Gamma(u) \neq \Gamma(v)$. Let $u',v' \in V(\tilde{G})$ such that $u \in \Gamma(u')$ and $v \in \Gamma(v')$. 
	We must have that $x, u', v'$ are pairwise different. But then, by (i) of Prop.~\ref{pr:BBtreeshortestpath}, $x$ is on a shortest path in $\tilde{G}$ that connects $u'$ and $v'$ and it has been added to $X_1$, as $u',v' \in C_1 \cup C_2$ by definition.  
	Now consider the case that $x \in V(B)$ for some block $B$ of $G$. Let $P$ be a shortest path in $G$ with endpoints $u$ and $v$ that contains $x$.
	If $u,v \in V(B)$, then the node $v_B \in V(\tilde{G})$ representing $B$ has been added to $Y$ in line~\ref{alg:initY} and processed in loop~\ref{alg:startfor}--\ref{alg:endfor}. In particular, $x$ is added to $X_{i+1}$ for some $i \geq 1$ because $u,v \in V(B) \cap X_i$ (cf. line~\ref{alg:updateX}).
	If at least one of $u,v$ is not a vertex of $B$, then let $u_\bot, v_\bot \in V(B)$ be the vertices on $P$ with the smallest distance to $u$ and $v$, respectively. 
	The definitions imply that $u_\bot, v_\bot \in V(\tilde{G})$.
	Furthermore, $u_\bot,v_\bot \in X_1$ by (i) of Prop.~\ref{pr:BBtreeshortestpath} and Lemma~\ref{lm:correctnesstreeclosure}.
	We are done if $x = u_\bot$ or $x = v_\bot$.
	Otherwise, $x$ is on a shortest path between $u_\bot$ and $v_\bot$ in $B$ and hence, as $u_\bot,v_\bot \in X_1 \subseteq X_i$, it is added to $X_{i+1}$ for some $i \geq 1$ in line~\ref{alg:updateX} for $v_B$, as $\closureblock$ is correct by Lemma~\ref{lm:correctnessblockclosure}. Hence, $\cc(X) \subseteq X_N$.
	
	Regarding the complexity, $\tilde{G}$ in line~\ref{alg:BBTree} can be computed in $\bigO{n}$ time~\cite{Horvath10} and, by Lemma~\ref{lm:correctnesstreeclosure}, the closure operator $\closuretree$ over $\tilde{G}$ (cf. line~\ref{alg:setC}) can be calculated also in $\bigO{n}$ time.
	Suppose $G$ contains $k$ blocks, say $B_1,\ldots,B_k$.
	Since by Lemma~\ref{lm:correctnessblockclosure}, the closure operator $\closureblock$ over $B_i$ (cf. line~\ref{alg:updateX}) can be computed in $\bigO{n_i f_i}$ time for all $i$, where $n_i = |V(B_i)|$ and $f_i =\facenumber(B_i)$, loop~\ref{alg:startfor}--\ref{alg:endfor} can be carried out in $\sum_i\bigO{n_i f_i} = \bigO{n f}$ time, as $\sum_i \bigO{n_i} = \bigO{n}$ and $f_i = \bigO{f}$. Thus, the total time of Alg.~\ref{alg:batchclosure} is $
	\bigO{n f}$, as claimed.
\end{proof}



\section{Experimental Results}
\label{sec:experiments}
Our experiments presented in this section are concerned with three properties of the proposed heuristic.
First, we evaluate Alg.~\ref{alg:outerplanarSampling} generating outerplanar spanning subgraphs for its \textit{runtime} and for the \textit{quality} of its output. The runtime results are compared also to those of standard algorithms generating \textit{spanning trees}. 
Second, we compare the runtime of our outerplanar closure algorithm (Alg.~\ref{alg:batchclosure}) to that of the na\"{\i}ve algorithm for outerplanar graphs (see Sect.~\ref{sec:preliminaries}). 
Third, using \textit{large real-world} networks~\cite{snapnets}, we empirically evaluate the approximation performance of our heuristic on the core-periphery decomposition~\cite{Sub18} problem.
For the implementation\footnote{The code is avaliable at \url{https://github.com/fseiffarth/GCoreApproximation}.} we used the C++-library {\sc Snap 6.0}~\cite{leskovec2016snap}. All experiments were conducted on an AMD Ryzen 9 3900X with 64GB RAM.

\subsection{Datasets}\label{sec:datasets}
\paragraph{\sc \ER I}
This dataset contains \textit{small} Erd\H{o}s-R\'enyi random graphs.
The size\footnote{We use only small graphs because testing how many edges can be added to the graph without destroying outerplanarity is in $\bigO{n m}$.} of the graphs vary, ranging from $n=100$ to $500$, with a step size of $100$ and for 10 different edge probabilities,  from $p=0.05$ with a step size of $0.01$. 
For each of the $50$ different configurations of $(n, p)$,  $100$ \textit{connected} \ER random graphs have been generated.

\paragraph{\sc \ER II}
This dataset contains also \ER connected random graphs with $10$ different sizes from $n=1{,}000$ with a step size of $1{,}000$ and with edge probabilities ranging from $p=0.006$ to $p=0.02$, with step size $0.002$. Below $p=0.006$, the graphs were too sparse for our purpose. For $n=10{,}000$ and $p=0.02$, the graphs contain around $1{,}000{,}000$ edges. For all configurations of $(n, p)$, $100$ \textit{connected} \ER random graphs have been generated.

\paragraph{\sc Large Real-World Networks}
This dataset contains 15 real-world networks from \cite{snapnets} (see \Cref{tab:graph_data}). 
In case of disconnected graphs, only their largest connected components were considered. 

\subsection{Sampling outerplanar spanning subgraphs}
By Thm.~\ref{thm:outerplanar}, Alg.~\ref{alg:outerplanarSampling} generates random spanning outerplanar graphs in $\bigO{m}$ time, which is, at the same time, the complexity of sampling random spanning trees (without the demand of uniform generation). 
The goal of our first experiments was to investigate the \textit{practical} \textit{time overhead} of Alg.~\ref{alg:outerplanarSampling} needed to generate spanning outerplanar graphs, instead of spanning trees.
The following algorithms have been considered for this purpose: 
(O1) is~Alg.~\ref{alg:outerplanarSampling}, 
(SBFS) implemented in {\sc Snap 6.0} generates a spanning BFS spanning tree, and (BFS) resp. (DFS) are our own implementations generating spanning BFS resp. DFS spanning trees.
We also consider (O2), which first calls (O1) and then calculates the BB-tree as well as the biconnected  components of the output outerplanar graphs returned by (O1). 
The reason of considering (O2) is that our closure algorithm requires these additional pieces of information as well.

We used the \ER II dataset for these experiments.
More precisely, for each of the $100$ random graphs for a particular value of $(n, p)$, we first generated a spanning outerplanar subgraph/tree with the above algorithms and compared their \textit{average} generation runtime.
The results are given in Fig.~\ref{fig:sampling_runtime}.
On the left-hand side we present the average runtime (in sec) of the algorithms per sample as a function of the number of edges in the input graphs.    
Our algorithms (O1) and (O2) scale linearly with the number of edges and are almost as fast as (BFS) and (DFS). Surprisingly, (SBFS) is even slower than (O2), though it generates spanning trees only.
As expected, (O2) is a bit slower than (O1) for the additional information it calculates. 
We will see that this is not a drawback w.r.t. the total time because the auxiliary structure generated by (O2) allows for a much faster closure computation. 
For graphs with around $10^6$ edges, (BFS) needs $0.38s$ per spanning \textit{tree}, while (O1) resp. (O2) $0.4s$ resp. $0.47s$ per spanning \textit{outerplanar} graph. 
If, however, we normalize the runtime by the number of edges in the spanning subgraph, (O1) is even \textit{faster} than (BFS) and (DFS) (see~Fig.~\ref{fig:sampling_runtime} (right)).
In summary, the time overhead of generating spanning outerplanar graphs instead of spanning trees is \textit{marginal}.

\begin{figure}[t]\centering
	\input{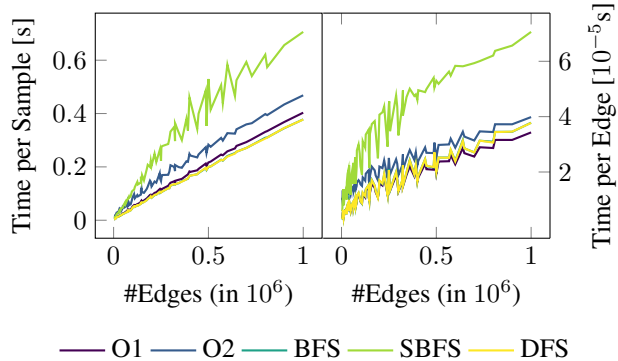}
	\caption{Average time of the algorithms generating spanning outerplanar subgraphs/trees for graphs up to $m=10^6$ in the \ER~II dataset. 
		(left) Average time (in sec) per input graph,     
		(right) average time (in $10^{-5}$ sec) per output edge. 
		\label{fig:sampling_runtime}}
\end{figure}

Besides runtime, we experimentally investigated the output outerplanar graphs of Alg.~\ref{alg:outerplanarSampling} for their \textit{maximality} and found that in all cases, they were at least  \textit{almost maximal}. These experiments were motivated by the fact that more edges in the spanning subgraph preserve more shortest paths.
For our experiments, we greedily added the edges of the input graph to the output of Alg.~\ref{alg:outerplanarSampling} which did not violate outerplanarity. 
Since this test needs $\bigO{m}$ time, we used the {\sc \ER I} dataset containing small graphs. 
The results are presented in \Cref{tab:quality1} for $n=500$ and different edge probabilities; we expect a similar behavior on larger graphs because the crucial factor for maximality is the graphs' density, and not their size. 
Around $18\%$ of the output outerplanar subgraphs were maximal, for the non-maximal outerplanar graphs only a few edges (in average less than $0.3\%$ of the edges) per graph were missing for maximality.

\begin{table}
	\begin{center}
		\begin{tabular}{crcccc}
			\toprule
			\multicolumn{1}{c}{Edge Prob.} &  \multicolumn{1}{c}{\#Edges} & \multicolumn{1}{c}{Maximal (\%)} & \multicolumn{1}{c}{Relative maximality (\%)} \\
			\midrule
			
			0.06 &  7,485 &          24 & 99.79 ($\pm$ 0.18) \\
			0.08 &  9,980 &          21 & 99.76 ($\pm$ 0.19) \\
			0.10 & 12,475 &          17 & 99.73 ($\pm$ 0.22) \\
			0.12 & 14,970 &          12 & 99.73 ($\pm$ 0.21) \\
			0.14 & 17,465 &          14 & 99.74 ($\pm$ 0.18) \\
			\bottomrule
		\end{tabular}
	\end{center}
	\caption{Percentage of maximal outerplanar graphs and closeness to a maximal outerplanar graph (in \%) of the output of Alg.~\ref{alg:outerplanarSampling} for graphs with $n=500$ in the {\sc \ER I} dataset. \label{tab:quality1}}
\end{table}

\subsection{Computing closures in outerplanar graphs}

In this section we empirically evaluate~Alg.~\ref{alg:batchclosure}.
More precisely, we take the output $G$ of (O1) for each graph in the {\sc \ER II} dataset. For each $G$, we first generate a random spanning tree $T$ of $G$ and construct then a possibly non-outerplanar graph $G'$ via adding $m-n+1$ edges to $T$. 
Thus, $G$ and $G'$ have the same number of vertices and edges. Fig.~\ref{fig:experiment_closure} (left) shows the average runtime needed to calculate the closures on $G$ and $G'$ for a random subset of $1\%$ of the vertices. 
(C1) is the na\"{\i}ve closure algorithm for the outerplanar graphs $G$ (i.e., it calculates the shortest paths between all pairs of input vertices). 
(C2) is our Alg.~\ref{alg:batchclosure} and (CGraph) is the na\"{\i}ve closure algorithm for the arbitrary graphs $G'$. 
Recall that the complexity of (CGraph) is $\bigO{nm}$, where $m = \bigO{n}$ by construction, it is $\bigO{n|X|}$  for (C1), where $|X|=n/100$, and $\bigO{n f}$ for (C2), which is independent of $|X|$.
The results are in accordance to these complexities. In particular, the closure computation on the arbitrary graphs $G'$ is slower by a factor up to $300$ than on the outerplanar graphs $G$ with (C1) and (C2) (see left of Fig.~\ref{fig:experiment_closure}). The right part of Fig.~\ref{fig:experiment_closure} is scaled down for (C1) and (C2). It clearly shows that (C2) (i.e., Alg.~\ref{alg:batchclosure}) is much faster in practice than the na\"{\i}ve algorithm (C1). In particular, (C2) seems to be the only out of the three algorithms which scales \textit{linearly} with the number of edges. This indicates that the face number $f$ in the time complexity $\bigO{n f}$ is \textit{negligible} in practice.

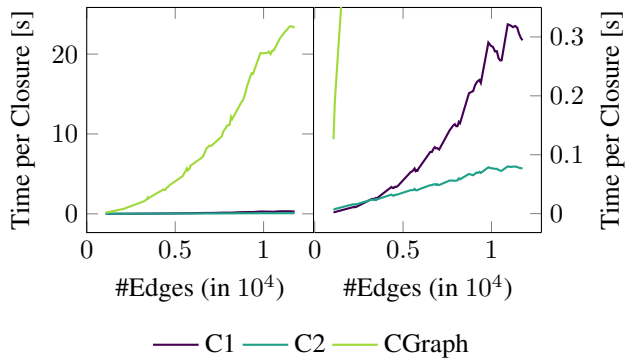
\begin{figure}[t]\centering
	\begin{tikzpicture}
	\pgfplotsset{scaled x ticks=false}
	\newcommand{\plota}{
		\begin{tikzpicture}[inner sep=3pt, trim axis right]			\definecolor{color0}{rgb}{0.267004,0.004874,0.329415}
		\definecolor{color1}{rgb}{0.275191,0.194905,0.496005}
		\definecolor{color2}{rgb}{0.212395,0.359683,0.55171}
		\definecolor{color3}{rgb}{0.153364,0.497,0.557724}
		\definecolor{color4}{rgb}{0.122312,0.633153,0.530398}
		\definecolor{color5}{rgb}{0.288921,0.758394,0.428426}
		\definecolor{color6}{rgb}{0.626579,0.854645,0.223353}
		\definecolor{color7}{rgb}{0.993248,0.906157,0.143936}
			\pgfkeys{/pgfplots/MyAxisStyle/.style={height=4cm,scale only axis, point meta=explicit symbolic}}
			\begin{axis}[MyAxisStyle,
				legend cell align={left},
				legend style={
					draw =none,
					legend columns=-1,
					fill opacity=0.8,
					draw opacity=1,
					text opacity=1,
					at={(0.25,-0.4)},
					anchor=north west,
				},
				width = \somewidth,
				height = \someheight,
				xtick = {0, 5000, 10000},
				xticklabels= {$0$, $0.5$, $1$},
				xlabel={\#Edges (in $10^4$)},
				ylabel={Time per Closure [s]},
				]
\addplot [thick, color0]
table [row sep=\\] {%
	1062 0.00277365\\1080 0.00259154\\1095 0.00270176\\1110 0.00270992\\1122 0.00271446\\1134 0.00277694\\1145 0.0028303\\1154 0.00281212\\2154 0.01088581\\2188 0.01112012\\2217 0.01078774\\2244 0.01092089\\2266 0.01108968\\2289 0.01115554\\2311 0.01129537\\2331 0.01119171\\3247 0.02512708\\3299 0.02498296\\3343 0.02435886\\3377 0.02511061\\3414 0.025702\\3446 0.02642011\\3478 0.02567526\\3508 0.02587978\\4343 0.04526831\\4407 0.04353642\\4466 0.04347214\\4517 0.04522486\\4563 0.04573711\\4606 0.04537688\\4644 0.04583192\\4686 0.04644323\\5437 0.07080115\\5519 0.07176801\\5590 0.07420675\\5651 0.07665106\\5709 0.07219358\\5762 0.07360934\\5817 0.07286329\\5863 0.07597306\\6535 0.1045907\\6634 0.10479654\\6711 0.10964386\\6793 0.11288435\\6858 0.11081293\\6925 0.11121876\\6984 0.1102162\\7040 0.10878749\\7631 0.14134253\\7742 0.14588696\\7838 0.14979661\\7929 0.15188051\\8007 0.14841842\\8081 0.1485488\\8155 0.15576374\\8221 0.15345338\\8729 0.20442611\\8854 0.20643171\\8964 0.20812841\\9069 0.21773507\\9155 0.22259103\\9240 0.22546078\\9324 0.22701083\\9399 0.22195506\\9826 0.29039725\\9965 0.28493601\\10084 0.28309738\\10201 0.27527924\\10307 0.27573057\\10401 0.26383853\\10493 0.26015417\\10575 0.26032487\\10921 0.32141361\\11076 0.31956557\\11214 0.31662635\\11339 0.31863943\\11454 0.31598801\\11558 0.30569728\\11660 0.29851292\\11757 0.29432884\\};
\addlegendentry{C1}
\addplot [thick, color4]
table [row sep=\\] {%
	1062 0.00772953\\1080 0.00730589\\1095 0.00740312\\1110 0.00742744\\1122 0.00761739\\1134 0.00778219\\1145 0.00777433\\1154 0.00791248\\2154 0.01572694\\2188 0.01587445\\2217 0.01524508\\2244 0.01526303\\2266 0.01569062\\2289 0.01571158\\2311 0.01597277\\2331 0.01556721\\3247 0.02405546\\3299 0.02357867\\3343 0.02297926\\3377 0.02327469\\3414 0.02374493\\3446 0.02442048\\3478 0.02326163\\3508 0.0240236\\4343 0.03235914\\4407 0.03077169\\4466 0.03031939\\4517 0.03157504\\4563 0.03158701\\4606 0.03206294\\4644 0.03215638\\4686 0.03255728\\5437 0.03987545\\5519 0.03975705\\5590 0.04042675\\5651 0.04201099\\5709 0.03983014\\5762 0.03967815\\5817 0.03978761\\5863 0.04083374\\6535 0.04831402\\6634 0.04821944\\6711 0.04967681\\6793 0.05000642\\6858 0.04889087\\6925 0.05013746\\6984 0.04855977\\7040 0.04858087\\7631 0.05620369\\7742 0.05614797\\7838 0.05611154\\7929 0.05572832\\8007 0.05589363\\8081 0.05659353\\8155 0.05991035\\8221 0.05830441\\8729 0.0673784\\8854 0.06614088\\8964 0.06521444\\9069 0.06771625\\9155 0.06978147\\9240 0.06917201\\9324 0.07090349\\9399 0.06912158\\9826 0.07874576\\9965 0.0771894\\10084 0.07677769\\10201 0.07601997\\10307 0.07612253\\10401 0.07424495\\10493 0.0729484\\10575 0.07268474\\10921 0.08018662\\11076 0.07932445\\11214 0.07984489\\11339 0.07991059\\11454 0.07932098\\11558 0.07778585\\11660 0.07723949\\11757 0.07655067\\};
\addlegendentry{C2}
\addplot [thick, color6]
table [row sep=\\] {%
	1062 0.12649521\\1080 0.14133317\\1095 0.15369992\\1110 0.16823018\\1122 0.17489252\\1134 0.19000469\\1145 0.19194738\\1154 0.19885362\\2154 0.64714365\\2188 0.71451183\\2217 0.72425368\\2244 0.75258818\\2266 0.77580612\\2289 0.8075026\\2311 0.82525339\\2331 0.83018897\\3247 1.57810716\\3299 1.67713417\\3343 1.71688755\\3377 1.79291582\\3414 1.88370934\\3446 2.04321029\\3478 1.92631251\\3508 1.9847398\\4343 3.06022113\\4407 2.97934064\\4466 3.07821401\\4517 3.26320307\\4563 3.34507578\\4606 3.40186911\\4644 3.49550429\\4686 3.58558793\\5437 4.75781575\\5519 5.12176799\\5590 5.35869887\\5651 5.69252451\\5709 5.42554931\\5762 5.55755405\\5817 5.66196517\\5863 6.02290965\\6535 7.08392546\\6634 7.45741262\\6711 7.90704963\\6793 8.25803627\\6858 8.35967119\\6925 8.5134729\\6984 8.53854883\\7040 8.50437115\\7631 9.67858599\\7742 10.40740368\\7838 10.66224718\\7929 11.03138112\\8007 11.0941039\\8081 11.14239949\\8155 12.15289668\\8221 11.84369941\\8729 13.90507182\\8854 14.30717277\\8964 15.09699181\\9069 16.0169262\\9155 16.55950372\\9240 17.13390101\\9324 17.59964317\\9399 17.52169508\\9826 20.13572143\\9965 20.08671117\\10084 20.12156983\\10201 20.14862604\\10307 20.31644538\\10401 19.98855458\\10493 20.4312959\\10575 20.44746677\\10921 21.95197382\\11076 22.4109848\\11214 22.79491912\\11339 23.08687342\\11454 23.42287573\\11558 23.48740086\\11660 23.41649923\\11757 23.2970613\\};
\addlegendentry{CGraph}
\end{axis}
		\end{tikzpicture}
	}
	
	\newcommand{\plotb}{
		\begin{tikzpicture}[inner sep=3pt, trim axis left]			\definecolor{color0}{rgb}{0.267004,0.004874,0.329415}
		\definecolor{color1}{rgb}{0.275191,0.194905,0.496005}
		\definecolor{color2}{rgb}{0.212395,0.359683,0.55171}
		\definecolor{color3}{rgb}{0.153364,0.497,0.557724}
		\definecolor{color4}{rgb}{0.122312,0.633153,0.530398}
		\definecolor{color5}{rgb}{0.288921,0.758394,0.428426}
		\definecolor{color6}{rgb}{0.626579,0.854645,0.223353}
		\definecolor{color7}{rgb}{0.993248,0.906157,0.143936}
			\pgfkeys{/pgfplots/MyAxisStyle/.style={height=4cm,scale only axis, point meta=explicit symbolic}}
			\begin{axis}[MyAxisStyle,
						legend style={
				draw =none,
				legend columns=-1,
				fill opacity=0,
				draw opacity=0,
				text opacity=0,
					at={(0.25,-0.4)},
				anchor=north west,
			},
				anchor=north west, xshift=1cm, axis y line*=right,
				width = \somewidth,
				height = \someheight,
				xlabel={\#Edges (in $10^4$)},
				ylabel={Time per Closure [s]},
				xtick = {0, 5000, 10000},
				xticklabels= {$0$, $0.5$, $1$},
				ymax = 0.35,
				]
\addplot [thick, color0]
table [row sep=\\] {%
	1062 0.00277365\\1080 0.00259154\\1095 0.00270176\\1110 0.00270992\\1122 0.00271446\\1134 0.00277694\\1145 0.0028303\\1154 0.00281212\\2154 0.01088581\\2188 0.01112012\\2217 0.01078774\\2244 0.01092089\\2266 0.01108968\\2289 0.01115554\\2311 0.01129537\\2331 0.01119171\\3247 0.02512708\\3299 0.02498296\\3343 0.02435886\\3377 0.02511061\\3414 0.025702\\3446 0.02642011\\3478 0.02567526\\3508 0.02587978\\4343 0.04526831\\4407 0.04353642\\4466 0.04347214\\4517 0.04522486\\4563 0.04573711\\4606 0.04537688\\4644 0.04583192\\4686 0.04644323\\5437 0.07080115\\5519 0.07176801\\5590 0.07420675\\5651 0.07665106\\5709 0.07219358\\5762 0.07360934\\5817 0.07286329\\5863 0.07597306\\6535 0.1045907\\6634 0.10479654\\6711 0.10964386\\6793 0.11288435\\6858 0.11081293\\6925 0.11121876\\6984 0.1102162\\7040 0.10878749\\7631 0.14134253\\7742 0.14588696\\7838 0.14979661\\7929 0.15188051\\8007 0.14841842\\8081 0.1485488\\8155 0.15576374\\8221 0.15345338\\8729 0.20442611\\8854 0.20643171\\8964 0.20812841\\9069 0.21773507\\9155 0.22259103\\9240 0.22546078\\9324 0.22701083\\9399 0.22195506\\9826 0.29039725\\9965 0.28493601\\10084 0.28309738\\10201 0.27527924\\10307 0.27573057\\10401 0.26383853\\10493 0.26015417\\10575 0.26032487\\10921 0.32141361\\11076 0.31956557\\11214 0.31662635\\11339 0.31863943\\11454 0.31598801\\11558 0.30569728\\11660 0.29851292\\11757 0.29432884\\};
\addlegendentry{Outerplanar Time}
\addplot [thick, color4]
table [row sep=\\] {%
	1062 0.00772953\\1080 0.00730589\\1095 0.00740312\\1110 0.00742744\\1122 0.00761739\\1134 0.00778219\\1145 0.00777433\\1154 0.00791248\\2154 0.01572694\\2188 0.01587445\\2217 0.01524508\\2244 0.01526303\\2266 0.01569062\\2289 0.01571158\\2311 0.01597277\\2331 0.01556721\\3247 0.02405546\\3299 0.02357867\\3343 0.02297926\\3377 0.02327469\\3414 0.02374493\\3446 0.02442048\\3478 0.02326163\\3508 0.0240236\\4343 0.03235914\\4407 0.03077169\\4466 0.03031939\\4517 0.03157504\\4563 0.03158701\\4606 0.03206294\\4644 0.03215638\\4686 0.03255728\\5437 0.03987545\\5519 0.03975705\\5590 0.04042675\\5651 0.04201099\\5709 0.03983014\\5762 0.03967815\\5817 0.03978761\\5863 0.04083374\\6535 0.04831402\\6634 0.04821944\\6711 0.04967681\\6793 0.05000642\\6858 0.04889087\\6925 0.05013746\\6984 0.04855977\\7040 0.04858087\\7631 0.05620369\\7742 0.05614797\\7838 0.05611154\\7929 0.05572832\\8007 0.05589363\\8081 0.05659353\\8155 0.05991035\\8221 0.05830441\\8729 0.0673784\\8854 0.06614088\\8964 0.06521444\\9069 0.06771625\\9155 0.06978147\\9240 0.06917201\\9324 0.07090349\\9399 0.06912158\\9826 0.07874576\\9965 0.0771894\\10084 0.07677769\\10201 0.07601997\\10307 0.07612253\\10401 0.07424495\\10493 0.0729484\\10575 0.07268474\\10921 0.08018662\\11076 0.07932445\\11214 0.07984489\\11339 0.07991059\\11454 0.07932098\\11558 0.07778585\\11660 0.07723949\\11757 0.07655067\\};
\addlegendentry{Outerplanar New Complete Time}
\addplot [thick, color6]
table [row sep=\\] {%
	1062 0.12649521\\1080 0.14133317\\1095 0.15369992\\1110 0.16823018\\1122 0.17489252\\1134 0.19000469\\1145 0.19194738\\1154 0.19885362\\2154 0.64714365\\};
\addlegendentry{CGraph}
\end{axis}	
		\end{tikzpicture}
	}
	
	\node (x) at (0,0) {};
	\node[inner sep=0pt] (a) at (5,0) {\plota};
	\node[inner sep=0pt, anchor=south west] (b) at ($ (a.south east) $) {\plotb};
\end{tikzpicture}
	\caption{(left) Closure runtimes for outerplanar graphs with the na\"{\i}ve alg. (C1) and with Alg.~\ref{alg:batchclosure} (C2) and for arbitrary graphs (CGraph), with the same number of nodes and edges. The generator set is a random subset of $1\%$ of the vertices.
		(right) Runtime scaled down for (C1) and (C2).
		\label{fig:experiment_closure}}
\end{figure}

\begin{table}
	\begin{center}
		\begin{tabular}{crccccc}
			\toprule
			\multicolumn{1}{c}{Edge Prob.} &   \multicolumn{1}{c}{\#Edges} &  \multicolumn{1}{c}{Avg. \#Output Edges} &   \multicolumn{1}{c}{Avg. Face Number} \\
			\midrule
			0.008 & 399,960 & 11,077.61 ($\pm$ 19.76) & 76.11 ($\pm$ 25.38) \\
			0.012 & 599,940 & 11,342.36 ($\pm$ 23.01) & 70.52 ($\pm$ 16.36) \\
			0.016 & 799,920 & 11,561.69 ($\pm$ 25.60) & 71.77 ($\pm$ 19.26) \\
			0.020 & 999,900 & 11,755.85 ($\pm$ 27.71) & 65.95 ($\pm$ 14.14) \\
			\bottomrule
		\end{tabular}
	\end{center}
	\caption{Output of~\Cref{alg:outerplanarSampling} on {\sc \ER II} random graphs with fixed size of $n=10^4$\label{tab:outerplanar_prop1}, averaged over $100$ samples}
\end{table}

This observation is supported by \Cref{tab:outerplanar_prop1} containing the average \textit{face number} of the generated spanning outerplanar subgraphs for the graphs with $n=10^4$ vertices in the {\sc\ER II} dataset. 
Somewhat surprisingly, the average face number does \textit{not} increase with the density.
Fig.~\ref{fig:face_evaluation} shows  the average face number as a function of the number of vertices (left) and the number of edges of the input graphs (right), where the colors represent different edge probabilities. The results indicate that in practice, the face number seems to be \textit{sublinear} in the graph size for fixed density (in our experiments, it was always less than $80$), justifying the \textit{better} runtime of our closure computation algorithm (see, again, the right of Fig.~\ref{fig:experiment_closure}).

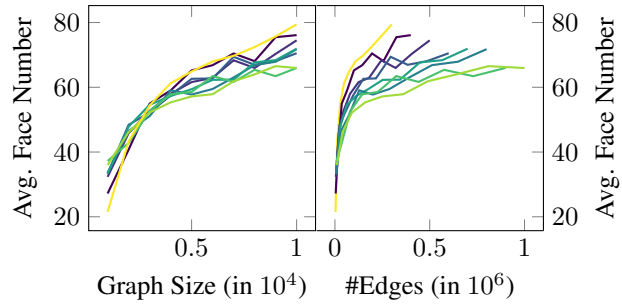
\begin{figure}\centering
	\begin{tikzpicture}
	\pgfplotsset{scaled x ticks=false}
	\newcommand{\plota}{
		\begin{tikzpicture}[inner sep=3pt, trim axis right]
			\definecolor{color0}{rgb}{0.267004,0.004874,0.329415}
			\definecolor{color1}{rgb}{0.275191,0.194905,0.496005}
			\definecolor{color2}{rgb}{0.212395,0.359683,0.55171}
			\definecolor{color3}{rgb}{0.153364,0.497,0.557724}
			\definecolor{color4}{rgb}{0.122312,0.633153,0.530398}
			\definecolor{color5}{rgb}{0.288921,0.758394,0.428426}
			\definecolor{color6}{rgb}{0.626579,0.854645,0.223353}
			\definecolor{color7}{rgb}{0.993248,0.906157,0.143936}
			\pgfkeys{/pgfplots/MyAxisStyle/.style={height=4cm,scale only axis, point meta=explicit symbolic}}
			\begin{axis}[MyAxisStyle,
				legend cell align={left},
			legend style={
	draw =none,
	legend columns=-1,
	fill opacity=0,
	draw opacity=0,
	text opacity=0,
	at={(0,0)},
	anchor=south west,},
				width = \somewidth,
				height = \someheight,
				x grid style={white!69.0196078431373!black},
				xlabel={Graph Size (in $10^4$)},
				ylabel={Avg. Face Number},
				xtick = {0, 5000, 10000},
				xticklabels= {$0$, $0.5$, $1$},
				]
\addplot [thick, color0]
table [row sep=\\] {%
	1000 27.19\\2000 40.73\\3000 54.94\\4000 59.02\\5000 65.22\\6000 66.82\\7000 70.44\\8000 68.07\\9000 75.47\\10000 76.11\\};
\addlegendentry{Avg. Face Number}
\addplot [thick, color1]
table [row sep=\\] {%
	1000 32.34\\2000 45.82\\3000 52.68\\4000 58.01\\5000 61.6\\6000 62.92\\7000 68.33\\8000 65.97\\9000 70.31\\10000 74.5\\};
\addlegendentry{Avg. Face Number}
\addplot [thick, color2]
table [row sep=\\] {%
	1000 33.32\\2000 48.39\\3000 52.24\\4000 57.48\\5000 62.62\\6000 62.73\\7000 69.41\\8000 66.94\\9000 68.09\\10000 70.52\\};
\addlegendentry{Avg. Face Number}
\addplot [thick, color3]
table [row sep=\\] {%
	1000 33.51\\2000 46.5\\3000 51.08\\4000 59.06\\5000 57.76\\6000 59.47\\7000 63.03\\8000 66.84\\9000 67.89\\10000 71.77\\};
\addlegendentry{Avg. Face Number}
\addplot [thick, color4]
table [row sep=\\] {%
	1000 33.34\\2000 47.82\\3000 54.61\\4000 57.29\\5000 59.41\\6000 62.3\\7000 62.35\\8000 67.68\\9000 68.38\\10000 71.92\\};
\addlegendentry{Avg. Face Number}
\addplot [thick, color5]
table [row sep=\\] {%
	1000 37.22\\2000 42.81\\3000 52.64\\4000 57.83\\5000 58.52\\6000 63.46\\7000 61.66\\8000 65.39\\9000 63.43\\10000 65.98\\};
\addlegendentry{Avg. Face Number}
\addplot [thick, color6]
table [row sep=\\] {%
	1000 35.98\\2000 46.16\\3000 52.21\\4000 55.34\\5000 57.17\\6000 57.86\\7000 61.86\\8000 63.98\\9000 66.53\\10000 65.95\\};
\addlegendentry{Avg. Face Number}
\addplot [thick, color7]
table [row sep=\\] {%
	1000 21.5\\2000 43.28\\3000 54.36\\4000 61.26\\5000 64.87\\6000 67.99\\7000 69.72\\8000 72.41\\9000 75.7\\10000 79.44\\};
\addlegendentry{Avg. Face Number}
\end{axis}
			
		\end{tikzpicture}
	}
	
\newcommand{\plotb}{
	\begin{tikzpicture}[inner sep=3pt, trim axis left]
		\definecolor{color0}{rgb}{0.267004,0.004874,0.329415}
		\definecolor{color1}{rgb}{0.275191,0.194905,0.496005}
		\definecolor{color2}{rgb}{0.212395,0.359683,0.55171}
		\definecolor{color3}{rgb}{0.153364,0.497,0.557724}
		\definecolor{color4}{rgb}{0.122312,0.633153,0.530398}
		\definecolor{color5}{rgb}{0.288921,0.758394,0.428426}
		\definecolor{color6}{rgb}{0.626579,0.854645,0.223353}
		\definecolor{color7}{rgb}{0.993248,0.906157,0.143936}
		\pgfkeys{/pgfplots/MyAxisStyle/.style={height=4cm,scale only axis, point meta=explicit symbolic}}
		\begin{axis}[MyAxisStyle,
			legend style={
			draw =none,
			legend columns=-1,
			fill opacity=0,
			draw opacity=0,
			text opacity=0,
			at={(0,0)},
			anchor=south west,},
				anchor=north west, xshift=1cm, axis y line*=right,
			width = \somewidth,
			height = \someheight,
			xtick = {0, 500000, 1000000},
			xticklabels= {$0$, $0.5$, $1$},
			xlabel={\#Edges (in $10^6$)},
			ylabel={Avg. Face Number},
			]
\addplot [thick, color7]
table [row sep=\\] {%
	2997 21.5\\11994 43.28\\26991 54.36\\47988 61.26\\74985 64.87\\107982 67.99\\146979 69.72\\191976 72.41\\242973 75.7\\299970 79.44\\};
\addlegendentry{Avg. Face Number}
\addplot [thick, color0]
table [row sep=\\] {%
	3996 27.19\\15992 40.73\\35988 54.94\\63984 59.02\\99980 65.22\\143976 66.82\\195972 70.44\\255968 68.07\\323964 75.47\\399960 76.11\\};
\addlegendentry{Avg. Face Number}
\addplot [thick, color1]
table [row sep=\\] {%
	4995 32.34\\19990 45.82\\44985 52.68\\79980 58.01\\124975 61.6\\179970 62.92\\244965 68.33\\319960 65.97\\404955 70.31\\499950 74.5\\};
\addlegendentry{Avg. Face Number}
\addplot [thick, color2]
table [row sep=\\] {%
	5994 33.32\\23988 48.39\\53982 52.24\\95976 57.48\\149970 62.62\\215964 62.73\\293958 69.41\\383952 66.94\\485946 68.09\\599940 70.52\\};
\addlegendentry{Avg. Face Number}
\addplot [thick, color4]
table [row sep=\\] {%
	6993 33.34\\27986 47.82\\62979 54.61\\111972 57.29\\174965 59.41\\251958 62.3\\342951 62.35\\447944 67.68\\566937 68.38\\699930 71.92\\};
\addlegendentry{Avg. Face Number}
\addplot [thick, color3]
table [row sep=\\] {%
	7992 33.51\\31984 46.5\\71976 51.08\\127968 59.06\\199960 57.76\\287952 59.47\\391944 63.03\\511936 66.84\\647928 67.89\\799920 71.77\\};
\addlegendentry{Avg. Face Number}
\addplot [thick, color5]
table [row sep=\\] {%
	8991 37.22\\35982 42.81\\80973 52.64\\143964 57.83\\224955 58.52\\323946 63.46\\440937 61.66\\575928 65.39\\728919 63.43\\899910 65.98\\};
\addlegendentry{Avg. Face Number}
\addplot [thick, color6]
table [row sep=\\] {%
	9990 35.98\\39980 46.16\\89970 52.21\\159960 55.34\\249950 57.17\\359940 57.86\\489930 61.86\\639920 63.98\\809910 66.53\\999900 65.95\\};
\addlegendentry{Avg. Face Number}
\end{axis}		
	\end{tikzpicture}
}

	\node (x) at (0,0) {};
	\node[inner sep=0pt] (a) at (5,0) {\plota};
	\node[inner sep=0pt, anchor=north west] (b) at ($ (a.north east) $) {\plotb};

\end{tikzpicture}
	\caption{Face numbers of outerplanar subgraphs generated by~\Cref{alg:outerplanarSampling} for the {\sc\ER II} dataset. (left) Average face number against the graph size. (right) Average face number against input edge number (colors depict edge probabilities).\label{fig:face_evaluation}}
\end{figure}

\begin{landscape}
\begin{table*}
	\begin{center}
		\begin{tabular}{lrrrrrr|rrc}
			\toprule
			Graph &    \multicolumn{1}{c}{Size}   &    \multicolumn{1}{c}{\#Edges}    &  \multicolumn{1}{c}{Density} &  \multicolumn{1}{c}{Size} &   \multicolumn{1}{c}{\#Edges} & \multicolumn{1}{c}{Time [s]} & \multicolumn{1}{|c}{Approx.} & \multicolumn{1}{c}{Time [s]} & \multicolumn{1}{c}{Jaccard sim.} \\
			&    \multicolumn{1}{c}{$n$}   &    \multicolumn{1}{c}{$m$}    &  \multicolumn{1}{c}{} &  \multicolumn{1}{c}{Core} &   \multicolumn{1}{c}{Core} & \multicolumn{1}{c}{Exact} & \multicolumn{1}{|c}{Core} & \multicolumn{1}{c}{Approx.} & \multicolumn{1}{c}{best ($l=5$, $t=1\%$)} \\
			\midrule
			com-Orkut & 3,072,441 & 117,185,083 & 2.5e-05 &        n.a. &         n.a. &      n.a. &    2,915,420 &        1.8e+04 &        n.a. \\
			soc-LiveJournal1 & 4,843,953 & 43,362,750 & 3.7e-06 &        n.a. &         n.a. &      n.a. &    3,018,149 &        8.7e+03 &        n.a. \\
			soc-pokec-relationships & 1,632,803 & 22,301,964 & 1.7e-05 &        n.a. &         n.a. &      n.a. &    1,390,297 &        6.5e+03 &       n.a. \\
com-youtube.ungraph & 1,134,890 &  2,987,624 & 4.6e-06 &   390,825 &  2,169,158 &        8.9e+05 &      338,654 &         2.2e+03 &               0.82 (0.73) \\
com-dblp.ungraph &   317,080 &  1,049,866 & 2.1e-05 &    90,077 &    438,265 &        7.0e+04 &       92,833 &         5.3e+02 &               \textbf{0.92} (0.87) \\
com-amazon.ungraph &   334,863 &    925,872 & 1.7e-05 &   216,109 &    643,075 &        2.2e+05 &      231,618 &         5.2e+02 &               0.88 (0.87) \\
Slashdot0902 &    82,168 &    582,533 & 1.7e-04 &    48,718 &    514,338 &        1.4e+04 &       45,558 &         1.6e+02 &              \textbf{0.92} (0.68) \\
Cit-HepPh &    34,401 &    420,828 & 7.1e-04 &    32,111 &    417,050 &        6.1e+03 &       32,309 &         9.6e+01 &               \textbf{0.99} (\textbf{0.97}) \\
Cit-HepTh &    27,400 &    352,059 & 9.4e-04 &    24,832 &    347,918 &        3.5e+03 &       25,049 &         7.7e+01 &               \textbf{0.98} (\textbf{0.98}) \\
CA-AstroPh &    17,903 &    197,031 & 1.2e-03 &     9,487 &    142,943 &        6.4e+02 &        9,522 &         3.0e+01 &               \textbf{0.95} (\textbf{0.93}) \\
CA-CondMat &    21,363 &     91,342 & 4.0e-04 &     8,603 &     49,682 &        4.0e+02 &        8,761 &         3.5e+01 &               \textbf{0.94} (\textbf{0.90}) \\
CA-HepPh &    11,204 &    117,649 & 1.9e-03 &     4,825 &     63,548 &        1.8e+02 &        4,804 &         1.8e+01 &               \textbf{0.93} (\textbf{0.91}) \\
Wiki-Vote &     7,066 &    100,736 & 4.0e-03 &     4,579 &     98,026 &        1.3e+02 &        4,452 &         1.5e+01 &               \textbf{0.97} (0.78) \\
CA-HepTh &     8,638 &     24,827 & 6.7e-04 &     3,605 &     14,161 &        4.6e+01 &        3,669 &         1.2e+01 &               \textbf{0.96} (\textbf{0.9}3) \\
CA-GrQc &     4,158 &     13,428 & 1.6e-03 &     1,336 &      5,036 &        7.0e+00 &        1,380 &         6.0e+00 &               \textbf{0.92} (0.88) \\
\bottomrule
		\end{tabular}
	\end{center}
	\caption{Large real-world networks with number of vertices ($n$), number of edges ($m$), density, number of vertices and edges in the core, time to calculate the exact core in seconds (or n.a. if it was not possible within 50 days), size of the approximated core, time to calculate the approximated core, and the Jaccard similarities of the exact and approximated cores obtained by grid search over $l$ and $t$, and for $l=5, t=1\%$ in brackets (values of at least 0.9 in bold). 
	The networks are sorted by $n \cdot m$.	
		\label{tab:graph_data}}
\end{table*}
\end{landscape}

\subsection{Core approximation in real-world networks}
\label{sec:core}
Finally, we present experiments concerning the approximation of \textit{cores} in \textit{large real-world} networks.
Following \cite{Sub18}, the core $\Core$ of a graph $G$ is defined by $\bigcap_{j=1}^i C_j$, where $i$ is the smallest integer satisfying $\bigcap_{j=1}^i C_j=\bigcap_{j=1}^{i+1} C_j$ and $C_j = \rho(X_j)$ is the closure of $X_j \subseteq V(G)$ containing $k >0$ vertices selected independently and uniformly at random. 
As a compromise between runtime and stability w.r.t. random effects, we choose $k=10$.  For each of the networks in Table~\ref{tab:graph_data}, the fixed point was reached after $i=3$ iterations. 

We used 15 networks from \cite{snapnets} in our experiments. 
The size ($n$) and order ($m$) of some of them are more than $1{,}000$ times larger than those in~\cite{Sub18}. Table~\ref{tab:graph_data} contains the size of the exact cores and the runtime of computing them. 
While the exact core of the $3$ largest networks could not be computed within 50 days, our algorithm produced the approximate cores in $5\opn{h}$ for these very large networks; in less than $40\opn{min}$ for all other graphs.

For the approximation, for each large network we generated $100$ spanning outerplanar subgraphs with Alg.~\ref{alg:outerplanarSampling} and calculated the closure of $l$ randomly chosen vertices on each of these outerplanar graphs with Alg.~\ref{alg:batchclosure}. 
Given the 100 closed sets in the outerplanar subgraphs obtained in this way, a vertex $v\in G$ was regarded as closed iff it was contained in at least $t\%$ of the closed sets. 
The approximate core $\tilde{\mathcal{C}}$ was then calculated in the same iterative way as the exact one, but with the approximate closed sets. 
We compared exact and approximate cores with each other using Jaccard similarity. 
The first value in the last column of Table~\ref{tab:graph_data} denotes the \textit{best} Jaccard similarity achieved via a grid search over $l \in\{5,\ldots,2000\}$ and $t \in \{1\%,\ldots,10\%\}$. 
We stress that using higher values of $l$ has no impact on the time complexity of our algorithm, as it depends on $n$ and the face number only (cf. Section~\ref{sec:closure}). 
The second value (in brackets) denotes the Jaccard similarity for the approximate core obtained for $l=5$ and $t=1\%$. 

For 12 out of the 15 graphs, we obtained a Jaccard similarity of around $0.8$ or more; for 9 even at least $0.9$.
As an example, in Fig.~\ref{fig:DegreeDistr} we show the exact core and periphery of the CA-HepTh network (see (a) and (b)) and their approximations (see (d) and (e)) for $l=5$ and $t = 1\%$ (see, also, Table~\ref{tab:graph_data}).
We also plot the degree distribution of the exact core (see (c)) and that of the approximate core (see (f)) obtained for these values. One can see that the two distributions are fairly similar to each other, by noting that the Jaccard similarity obtained for $l=5$ and $t=1\%$ was $0.93$ (see Table~\ref{tab:graph_data}). A similar behavior could be observed for the other networks as well.  

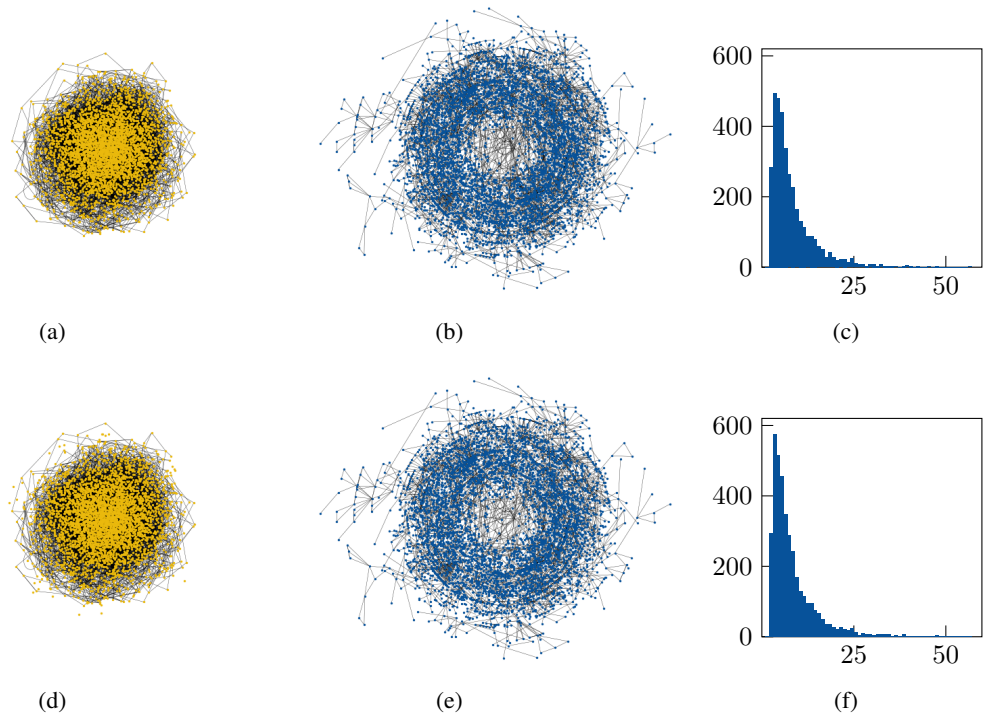
\begin{figure}
	\captionsetup[subfigure]{justification=centering}
	\centering
	\newcommand{\figscale}{0.35}
	\begin{subfigure}[b]{0.3\linewidth}
		\centering
		\includegraphics[scale=\figscale]{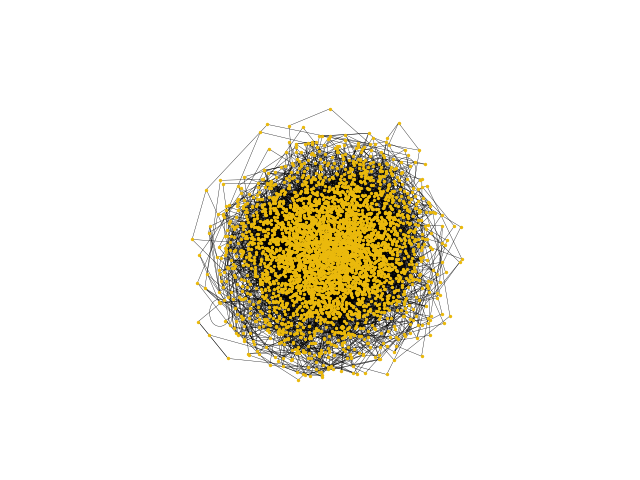}
		\caption{}
		\label{fig:Core_}
	\end{subfigure}
	\hfill
	\begin{subfigure}[b]{0.3\linewidth}
		\centering
		\includegraphics[scale=\figscale]{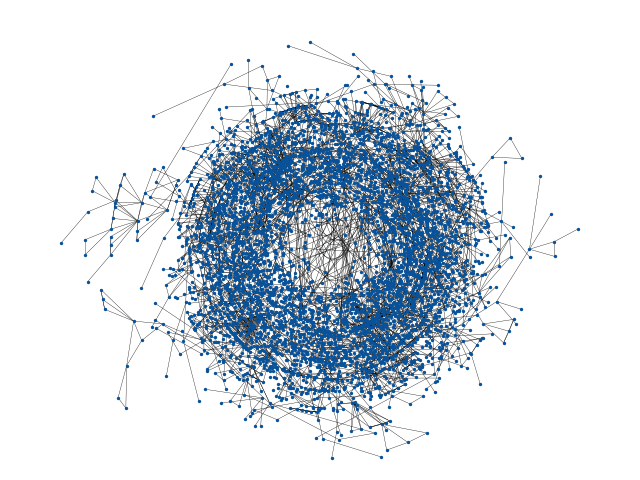}
		\caption{}
		\label{fig:Periphery_}
	\end{subfigure}
	\hfill
	\begin{subfigure}[b]{0.3\linewidth}
		\centering
\begin{tikzpicture}
	
	\definecolor{darkgray176}{RGB}{176,176,176}
	\definecolor{teal782154}{RGB}{7,82,154}
	
	\begin{axis}[
		tick align=inside,
		tick pos=left,
		x grid style={darkgray176},
		xmin= 0, xmax=59.8,
		xtick={25, 50},
		xtick style={color=black},
		y grid style={darkgray176},
		ymin= 0, ymax=620,
ytick={0, 200, 400, 600},
		ytick style={color=black},
		width = \somewidth*1.5,
		height = \someheight*1.5,
		]
		\draw[draw=none,fill=teal782154] (axis cs:1,0) rectangle (axis cs:2,0);
		\draw[draw=none,fill=teal782154] (axis cs:2,0) rectangle (axis cs:3,284);
		\draw[draw=none,fill=teal782154] (axis cs:3,0) rectangle (axis cs:4,495);
		\draw[draw=none,fill=teal782154] (axis cs:4,0) rectangle (axis cs:5,481);
		\draw[draw=none,fill=teal782154] (axis cs:5,0) rectangle (axis cs:6,441);
		\draw[draw=none,fill=teal782154] (axis cs:6,0) rectangle (axis cs:7,338);
		\draw[draw=none,fill=teal782154] (axis cs:7,0) rectangle (axis cs:8,264);
		\draw[draw=none,fill=teal782154] (axis cs:8,0) rectangle (axis cs:9,228);
		\draw[draw=none,fill=teal782154] (axis cs:9,0) rectangle (axis cs:10,164);
		\draw[draw=none,fill=teal782154] (axis cs:10,0) rectangle (axis cs:11,131);
		\draw[draw=none,fill=teal782154] (axis cs:11,0) rectangle (axis cs:12,114);
		\draw[draw=none,fill=teal782154] (axis cs:12,0) rectangle (axis cs:13,87);
		\draw[draw=none,fill=teal782154] (axis cs:13,0) rectangle (axis cs:14,89);
		\draw[draw=none,fill=teal782154] (axis cs:14,0) rectangle (axis cs:15,80);
		\draw[draw=none,fill=teal782154] (axis cs:15,0) rectangle (axis cs:16,58);
		\draw[draw=none,fill=teal782154] (axis cs:16,0) rectangle (axis cs:17,50);
		\draw[draw=none,fill=teal782154] (axis cs:17,0) rectangle (axis cs:18,27);
		\draw[draw=none,fill=teal782154] (axis cs:18,0) rectangle (axis cs:19,42);
		\draw[draw=none,fill=teal782154] (axis cs:19,0) rectangle (axis cs:20,27);
		\draw[draw=none,fill=teal782154] (axis cs:20,0) rectangle (axis cs:21,21);
		\draw[draw=none,fill=teal782154] (axis cs:21,0) rectangle (axis cs:22,22);
		\draw[draw=none,fill=teal782154] (axis cs:22,0) rectangle (axis cs:23,23);
		\draw[draw=none,fill=teal782154] (axis cs:23,0) rectangle (axis cs:24,14);
		\draw[draw=none,fill=teal782154] (axis cs:24,0) rectangle (axis cs:25,24);
		\draw[draw=none,fill=teal782154] (axis cs:25,0) rectangle (axis cs:26,10);
		\draw[draw=none,fill=teal782154] (axis cs:26,0) rectangle (axis cs:27,8);
		\draw[draw=none,fill=teal782154] (axis cs:27,0) rectangle (axis cs:28,8);
		\draw[draw=none,fill=teal782154] (axis cs:28,0) rectangle (axis cs:29,4);
		\draw[draw=none,fill=teal782154] (axis cs:29,0) rectangle (axis cs:30,7);
		\draw[draw=none,fill=teal782154] (axis cs:30,0) rectangle (axis cs:31,9);
		\draw[draw=none,fill=teal782154] (axis cs:31,0) rectangle (axis cs:32,4);
		\draw[draw=none,fill=teal782154] (axis cs:32,0) rectangle (axis cs:33,7);
		\draw[draw=none,fill=teal782154] (axis cs:33,0) rectangle (axis cs:34,4);
		\draw[draw=none,fill=teal782154] (axis cs:34,0) rectangle (axis cs:35,4);
		\draw[draw=none,fill=teal782154] (axis cs:35,0) rectangle (axis cs:36,4);
		\draw[draw=none,fill=teal782154] (axis cs:36,0) rectangle (axis cs:37,2);
		\draw[draw=none,fill=teal782154] (axis cs:37,0) rectangle (axis cs:38,1);
		\draw[draw=none,fill=teal782154] (axis cs:38,0) rectangle (axis cs:39,4);
		\draw[draw=none,fill=teal782154] (axis cs:39,0) rectangle (axis cs:40,5);
		\draw[draw=none,fill=teal782154] (axis cs:40,0) rectangle (axis cs:41,3);
		\draw[draw=none,fill=teal782154] (axis cs:41,0) rectangle (axis cs:42,1);
		\draw[draw=none,fill=teal782154] (axis cs:42,0) rectangle (axis cs:43,3);
		\draw[draw=none,fill=teal782154] (axis cs:43,0) rectangle (axis cs:44,1);
		\draw[draw=none,fill=teal782154] (axis cs:44,0) rectangle (axis cs:45,0);
		\draw[draw=none,fill=teal782154] (axis cs:45,0) rectangle (axis cs:46,2);
		\draw[draw=none,fill=teal782154] (axis cs:46,0) rectangle (axis cs:47,1);
		\draw[draw=none,fill=teal782154] (axis cs:47,0) rectangle (axis cs:48,3);
		\draw[draw=none,fill=teal782154] (axis cs:48,0) rectangle (axis cs:49,1);
		\draw[draw=none,fill=teal782154] (axis cs:49,0) rectangle (axis cs:50,1);
		\draw[draw=none,fill=teal782154] (axis cs:50,0) rectangle (axis cs:51,0);
		\draw[draw=none,fill=teal782154] (axis cs:51,0) rectangle (axis cs:52,1);
		\draw[draw=none,fill=teal782154] (axis cs:52,0) rectangle (axis cs:53,0);
		\draw[draw=none,fill=teal782154] (axis cs:53,0) rectangle (axis cs:54,0);
		\draw[draw=none,fill=teal782154] (axis cs:54,0) rectangle (axis cs:55,0);
		\draw[draw=none,fill=teal782154] (axis cs:55,0) rectangle (axis cs:56,1);
		\draw[draw=none,fill=teal782154] (axis cs:56,0) rectangle (axis cs:57,2);
	\end{axis}
	
\end{tikzpicture}
		\caption{}
		\label{fig:DegreeDist_}
	\end{subfigure}
\\
	\begin{subfigure}[b]{0.3\linewidth}
	\centering
	\includegraphics[scale=\figscale]{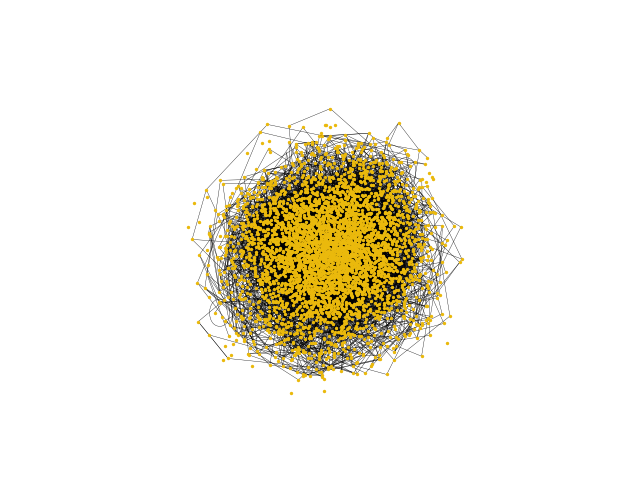}
	\caption{}
	\label{fig:Core_Approx}
\end{subfigure}
\hfill
\begin{subfigure}[b]{0.3\linewidth}
	\centering
	\includegraphics[scale=\figscale]{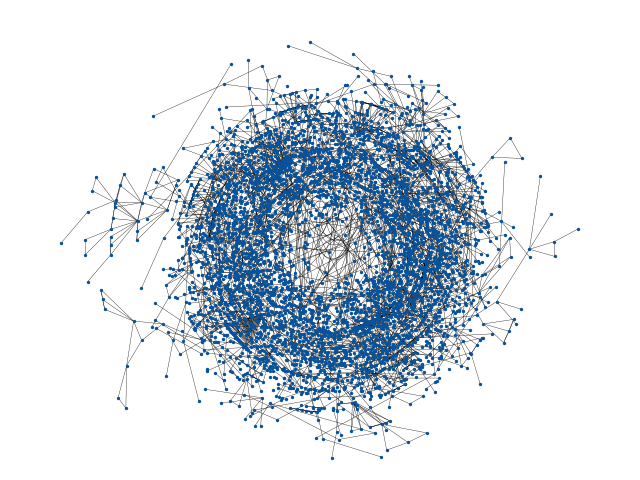}
	\caption{}
	\label{fig:Periphery_Approx}
\end{subfigure}
\hfill
\begin{subfigure}[b]{0.3\linewidth}
	\centering
\begin{tikzpicture}
	
	\definecolor{darkgray176}{RGB}{176,176,176}
	\definecolor{teal782154}{RGB}{7,82,154}
	
	\begin{axis}[
		tick align=inside,
		tick pos=left,
		x grid style={darkgray176},
		xmin= 0, xmax=59.8,
		xtick={25, 50},
		xtick style={color=black},
		y grid style={darkgray176},
		ymin= 0, ymax=620,
		ytick={0, 200, 400, 600},
		ytick style={color=black},
		width = \somewidth*1.5,
		height = \someheight*1.5,
		]
		\draw[draw=none,fill=teal782154] (axis cs:1,0) rectangle (axis cs:2,0);
		\draw[draw=none,fill=teal782154] (axis cs:2,0) rectangle (axis cs:3,293);
		\draw[draw=none,fill=teal782154] (axis cs:3,0) rectangle (axis cs:4,574);
		\draw[draw=none,fill=teal782154] (axis cs:4,0) rectangle (axis cs:5,514);
		\draw[draw=none,fill=teal782154] (axis cs:5,0) rectangle (axis cs:6,457);
		\draw[draw=none,fill=teal782154] (axis cs:6,0) rectangle (axis cs:7,347);
		\draw[draw=none,fill=teal782154] (axis cs:7,0) rectangle (axis cs:8,289);
		\draw[draw=none,fill=teal782154] (axis cs:8,0) rectangle (axis cs:9,244);
		\draw[draw=none,fill=teal782154] (axis cs:9,0) rectangle (axis cs:10,168);
		\draw[draw=none,fill=teal782154] (axis cs:10,0) rectangle (axis cs:11,130);
		\draw[draw=none,fill=teal782154] (axis cs:11,0) rectangle (axis cs:12,116);
		\draw[draw=none,fill=teal782154] (axis cs:12,0) rectangle (axis cs:13,95);
		\draw[draw=none,fill=teal782154] (axis cs:13,0) rectangle (axis cs:14,94);
		\draw[draw=none,fill=teal782154] (axis cs:14,0) rectangle (axis cs:15,74);
		\draw[draw=none,fill=teal782154] (axis cs:15,0) rectangle (axis cs:16,68);
		\draw[draw=none,fill=teal782154] (axis cs:16,0) rectangle (axis cs:17,51);
		\draw[draw=none,fill=teal782154] (axis cs:17,0) rectangle (axis cs:18,35);
		\draw[draw=none,fill=teal782154] (axis cs:18,0) rectangle (axis cs:19,36);
		\draw[draw=none,fill=teal782154] (axis cs:19,0) rectangle (axis cs:20,28);
		\draw[draw=none,fill=teal782154] (axis cs:20,0) rectangle (axis cs:21,20);
		\draw[draw=none,fill=teal782154] (axis cs:21,0) rectangle (axis cs:22,28);
		\draw[draw=none,fill=teal782154] (axis cs:22,0) rectangle (axis cs:23,20);
		\draw[draw=none,fill=teal782154] (axis cs:23,0) rectangle (axis cs:24,18);
		\draw[draw=none,fill=teal782154] (axis cs:24,0) rectangle (axis cs:25,23);
		\draw[draw=none,fill=teal782154] (axis cs:25,0) rectangle (axis cs:26,13);
		\draw[draw=none,fill=teal782154] (axis cs:26,0) rectangle (axis cs:27,3);
		\draw[draw=none,fill=teal782154] (axis cs:27,0) rectangle (axis cs:28,10);
		\draw[draw=none,fill=teal782154] (axis cs:28,0) rectangle (axis cs:29,6);
		\draw[draw=none,fill=teal782154] (axis cs:29,0) rectangle (axis cs:30,6);
		\draw[draw=none,fill=teal782154] (axis cs:30,0) rectangle (axis cs:31,5);
		\draw[draw=none,fill=teal782154] (axis cs:31,0) rectangle (axis cs:32,6);
		\draw[draw=none,fill=teal782154] (axis cs:32,0) rectangle (axis cs:33,6);
		\draw[draw=none,fill=teal782154] (axis cs:33,0) rectangle (axis cs:34,7);
		\draw[draw=none,fill=teal782154] (axis cs:34,0) rectangle (axis cs:35,6);
		\draw[draw=none,fill=teal782154] (axis cs:35,0) rectangle (axis cs:36,1);
		\draw[draw=none,fill=teal782154] (axis cs:36,0) rectangle (axis cs:37,3);
		\draw[draw=none,fill=teal782154] (axis cs:37,0) rectangle (axis cs:38,1);
		\draw[draw=none,fill=teal782154] (axis cs:38,0) rectangle (axis cs:39,8);
		\draw[draw=none,fill=teal782154] (axis cs:39,0) rectangle (axis cs:40,2);
		\draw[draw=none,fill=teal782154] (axis cs:40,0) rectangle (axis cs:41,2);
		\draw[draw=none,fill=teal782154] (axis cs:41,0) rectangle (axis cs:42,2);
		\draw[draw=none,fill=teal782154] (axis cs:42,0) rectangle (axis cs:43,2);
		\draw[draw=none,fill=teal782154] (axis cs:43,0) rectangle (axis cs:44,1);
		\draw[draw=none,fill=teal782154] (axis cs:44,0) rectangle (axis cs:45,1);
		\draw[draw=none,fill=teal782154] (axis cs:45,0) rectangle (axis cs:46,2);
		\draw[draw=none,fill=teal782154] (axis cs:46,0) rectangle (axis cs:47,0);
		\draw[draw=none,fill=teal782154] (axis cs:47,0) rectangle (axis cs:48,3);
		\draw[draw=none,fill=teal782154] (axis cs:48,0) rectangle (axis cs:49,1);
		\draw[draw=none,fill=teal782154] (axis cs:49,0) rectangle (axis cs:50,1);
		\draw[draw=none,fill=teal782154] (axis cs:50,0) rectangle (axis cs:51,1);
		\draw[draw=none,fill=teal782154] (axis cs:51,0) rectangle (axis cs:52,0);
		\draw[draw=none,fill=teal782154] (axis cs:52,0) rectangle (axis cs:53,0);
		\draw[draw=none,fill=teal782154] (axis cs:53,0) rectangle (axis cs:54,0);
		\draw[draw=none,fill=teal782154] (axis cs:54,0) rectangle (axis cs:55,1);
		\draw[draw=none,fill=teal782154] (axis cs:55,0) rectangle (axis cs:56,0);
		\draw[draw=none,fill=teal782154] (axis cs:56,0) rectangle (axis cs:57,2);
	\end{axis}
	
\end{tikzpicture}
	\caption{}
	\label{fig:DegreeDist_Approx}
\end{subfigure}
	\caption{CA-HepTh network, its exact (a) core, (b) periphery, (c) degree distribution of the core and its approximated (d)  core, (e) periphery,  (f) degree distribution of the approx. core.}
	\label{fig:DegreeDistr}
\end{figure}


\section{Concluding Remarks}
\label{sec:conclusion}

Our experimental results clearly demonstrate that the presence of \textit{cyclic} edges in the spanning subgraphs is essential for a close approximation of the geodesic convex hull. 
Thus, it is natural to ask whether further graph classes \textit{beyond} forests can also be considered for spanning subgraphs.
Such a graph class should fulfill at least two properties: 
(i) A (potentially maximal) spanning subgraph from this class could be generated in time \textit{linear} in the order of the input graph and (ii) for the graphs in this class, the preclosure of any set vertices should be its closure at the same time (cf. Thm.~\ref{thm:outerplanar_preclosure} in Sect.~\ref{sec:preliminaries}).  
This second condition indicates that the graphs in the class should be $K_{2,3}$-free (w.r.t. forbidden minor).
A somewhat related question is whether the algorithm presented in Sect.~\ref{sec:sampling} can be modified in a way that it returns a \textit{maximal} spanning outerplanar graph, preserving at the same time the time complexity of Alg.~\ref{alg:outerplanarSampling}.
For example, is it possible to utilize the \textit{degree distribution} of the input graph in the selection of the back edges in a way that the output outerplanar graph is always maximal?

Although our primary focus in this work was on an effective approximation of geodesic convex hulls in large graphs, 
the results of Section~\ref{sec:core} raise some interesting questions towards this direction. For example, we are investigating whether it is possible to approximate the set of nodes with the highest \textit{betweenness centrality} in large networks by that in their approximate cores?
 
Our empirical results concerning core approximation in large real-world networks have been obtained for relatively small sets of generator elements and for low frequency thresholds. 
The choice of these two parameters 
seem crucial for a close approximation (see Table~\ref{tab:graph_data}).
The related question is how to select them, especially in case of large networks? 
\textit{Sampling} seems a natural way, the question is whether it is possible to utilize the structure of the network at hand during sampling?   
Last but not least, it would be interesting to \textit{systematically} study further types of random as well as large real-world networks for their core-periphery decomposition.

\bibliographystyle{plain}
\bibliography{references}

\end{document}